\documentclass[prx,aps,twocolumn,notitlepage,superscriptaddress,showpacs,nofootinbib,10pt]{revtex4-2}

\usepackage[T2A]{fontenc}
\usepackage{enumerate,appendix}
\usepackage{amsmath,amsthm, amssymb,commath}
\usepackage{graphicx}
\usepackage[colorlinks]{hyperref}
\hypersetup{
	colorlinks = true,
	urlcolor = {blue},
	citecolor = {blue},
	linkcolor= {blue}
} 
\usepackage{subcaption}
\usepackage{verbatim}
\usepackage[justification=RaggedRight]{caption}

\usepackage{optidef} 
\usepackage{xspace}
\usepackage{etoolbox}

\newtoggle{science}
\togglefalse{science}

\newtheorem{theorem}{Theorem} 
\newtheorem{lemma}[theorem]{Lemma}

\def\>{\rangle} 
\def\<{\langle}

\begin{document}

\title{A theory of quantum error correction for permutation-invariant codes} 
\author{Yingkai Ouyang}
\email{y.ouyang@sheffield.ac.uk}
\affiliation{School of Mathematical and Physical Sciences, University of Sheffield, Sheffield, S3 7RH, United Kingdom} 

\author{Gavin K. Brennen}
\email{gavin.brennen@mq.edu.au}
\affiliation{Center for Engineered Quantum Systems, Dept. of Physics \& Astronomy, Macquarie University, 2109 NSW, Australia}

\affiliation{BTQ Technologies, 16-104 555 Burrard Street, Vancouver, British Columbia, Canada V7X 1M8}

\begin{abstract} 
%   
%Permutation invariant (PI) codes are non-stabilizer quantum error correction codes that offer some attractive features including:  high performance for correcting amplitude damping errors, designer bit flip/phase flip distance, and admissibility of  transversal gates outside the Clifford hierarchy. 
We present for the first time a general theory of error correction for permutation invariant (PI) codes.  Using representation theory of the symmetric group we construct efficient algorithms that can correct any correctible error on any PI code. These algorithms involve measurements of total angular momentum, quantum Schur transforms or logical state teleportations, and geometric phase gates. For erasure errors, or more generally deletion errors, on certain PI codes, we give a simpler quantum error correction algorithm.
\end{abstract}

\maketitle 

\section{Introduction}

Recent advances in quantum error correction (QEC) have resulted in increasingly sophisticated codes that promise pathways to fault-tolerant quantum processing \cite{pk22,bravyi2024high}. In concert, experimental progress has been rapid with demonstrations of beyond break-even performance in storing encoded quantum memory \cite{Acharya2025-ba} and logical non-Clifford gates \cite{dck4-x9c2,dasu2025breakingmagicdemonstrationhighfidelity} and entangling Clifford gates \cite{Bluvstein2024-ht} on encoded qubits.
To date most effort has been focused on stabilizer codes wherein the code space is stabilized by parity checks on qubits in conjugate bases. While relatively simple to work with, stabilizer codes do have some limitations. For example, they perform less well than non-stabilizer codes at correcting non-Pauli errors like amplitude damping \cite{PhysRevA.89.022316} and they cannot correct for deletion errors \cite{leahy2019quantum} which are qubit erasures at unknown qubit locations.
%The former issue, a necessary consequence of the Eastin-Knill theorem \cite{PhysRevLett.102.110502} which disallows a universal transversal gate set within any single code, is usually obviated by code switching to another stabilizer code \cite{PhysRevLett.113.080501} or using magic state preparation \cite{knill2004}.
%However, the obtainable transversal gates are still restricted to be within the Clifford hierarchy. 
The latter type of error is relevant in architectures like trapped Rydberg atom arrays where leakage can be converted to loss blindly, but detecting the location of loss requires more resources \cite{PRXQuantum.5.040343}. Also, the set of admissible transversal logical operators is limited to the Clifford hierarchy and the resource overhead to perform fault tolerant small angle rotation gates with stabilizer codes via code switching is substantial \cite{JainAlbert2025}.  

In this work, we develop a theory of quantum error correction for permutation invariant (PI) codes \cite{Rus00,PoR04,ouyang2014permutation,ouyang2015permutation,OUYANG201743,movassagh2020constructing,ouyang2025measurement} which are non-stabilizer QEC codes with a codespace comprised of strictly permutation symmetric states. There are a variety of families of such codes which share several attractive features.
First, the permutation symmetry of PI codes makes the location of single particle errors irrelevant. This can make correction of amplitude errors simpler relative to stabilizer codes and also allows for correcting for the aforementioned deletion errors, with correctability arising directly from their distance property \cite{HagiwaraISIT2020,ouyang2021permutation}.
Second, PI codes have been demonstrated to also correct against insertion errors \cite{shibayama2021equivalence,bulled2025angular}, where insertion errors models the addition of unknown qubits at unknown locations to the quantum system \cite{leahy2019quantum}, 
and are challenging for conventional stabilizer codes to handle.
Third, they have applications such as in QEC-enhanced quantum sensing \cite{ouyang2019robust,ouyang2022finite}, transversal non-Clifford gates with exotic angles of rotation that are activated by measurement-free code-switching \cite{ouyang2025measurement}, and quantum storage \cite{ouyang2019quantum}.
Fourth, PI codes been recently demonstrated to exhibit deep connections with emerging families of quantum codes, such as spin codes and bosonic codes \cite{omanakuttan2026holsteinprimakoffspincodes,aydin2025quantum,kubischta2023not,kubischta2025intrinsic,gross-code,BinomialCodes2016}.

From a controllability perspective, PI codes have the advantage that full state and unitary synthesis in the symmetric state space is achievable using global fields that collectively rotate the qubits together with a uniform coupling to a bosonic mode that mediates entangling gates \cite{PhysRevResearch.7.L022072,PhysRevA.110.062610,johnsson2020geometric}. The mode can be of motional or electro-magnetic type and this feature can considerably simplify implementations in platforms like trapped ions or Rydberg atom arrays where localized addressability, especially in large registers, is challenging due to crosstalk errors \cite{ParradoRodriguez2021crosstalk,MorgadoWhitlock2021}. 
%The ease of preparing symmetric states in actual physical systems together with the possibility of performing QEC on these states makes them a promising candidate for demonstrating quantum error correction in a variety of physical systems, such as in ion-traps, atomic gases, NV centers and ultracold atoms.

Although PI codes are attractive candidates to consider for QEC, a full theory for how to do so remains lacking.
In particular, while the optimal recovery map that performs the QEC for PI codes exists, and its Kraus operators can be written down \cite{KnL97}, it was not known how one could concretely implement these recovery operations using a sequence of simple operations.  
% PI codes, being very different from most QEC codes which are stabilizer codes, necessitates development of a new QEC theory that resolves this issue.
In this paper, we complete the theory of QEC for PI codes. Applying tools from the representation theory of the symmetric group specialised to qubits \cite{harrow2013church,havlivcek2018quantum}, we devise simple QEC operations for PI codes.
We present for the first time, efficient and near-term decoding algorithms for implementing QEC on PI codes, be it a quantum circuit model, or with near-term quantum control.
%These operations require only measurements of total angular momenta, quantum Schur transforms, and geometric phase gates. 

Our method follows a two-stage procedure.
Stage 1 projects the corrupted state onto the irreducible representations of the symmetric group by measuring the total angular momentum of subsets of qubits. The error-syndrome extracted at this point corresponds to a standard Young tableaux.
Stage 2 uses the standard Young tableaux syndrome to complete the decoding; it performs an adaptive quantum algorithm based on the standard Young tableaux syndrome that brings the state back to the codespace. This algorithm comprises partial inverse quantum Schur transformations \cite{BCH-PRL-2006-schur} and amplitude rebalancing steps.
For deletion errors with certain PI codes, the recovery operation is even simpler and does not use the aforementioned two-stage quantum error correction procedure; we measure in the modulo Dicke basis, and apply geometric phase gates (GPGs).

We propose two different flavors of our decoding protocol, one using a mix of GPGs \cite{johnsson2020geometric} and the usual Clifford + $T$ gateset, and another using fully GPGs. In both schemes, our decoding scheme is efficient, in the sense that it has a worst-case complexity that is at most linear in the number of qubits for both schemes. 
%GPGs only require (1) a switchable dispersive or linear coupling of qubits uniformly to a catalytic bosonic mode, (2) displacement operations, and (3) balanced homodyne measurements on the mode.
%GPGs require only four native operations. The first operation is the initialization of the bosonic mode, which is achievable with the use of a laser.
%The second operation is a coherent on-off dispersive coupling of all qubits to that mode. This is routinely demonstrated in cavity QED architectures \cite{eickbusch2022fast,wang2024dispersive}. Third, we need displacement of the mode, and fourth, homodyne detection, both of which have mature implementations across an array of platforms \cite{PhysRevLett.94.113601,PhysRevLett.130.143004}. 
The benefit of moving beyond the usual Clifford + $T$ paradigm to fully using GPGs is the elimination of most of the need for individual qubit addressability, which exploits readily available bosonic manipulations. This offers a hardware-efficient route to realizing our protocol in near-term devices. 
  
\iftoggle{science}
{\section{Results}
\subsection{PI codes}}
{\section{PI codes}}

Symmetric states, invariant under any permutation of their underlying particles, are superpositions of the Dicke states  
$| D^N_w\>  ={\binom N w}^{-1/2}
\sum_{\substack{ x_1, \dots, x_N \in \{0,1\} \\ x_1+ \dots + x_N = w}}
|x_1, \dots, x_N\>,$ 
of weights $w = 0,\dots, N$.  
PI codes are QEC codes comprising of symmetric states, and are generally non-additive and therefore non-stabilizer. A notable family of PI codes are the $s$-shifted gnu codes \cite{ouyang2021permutation,OUYANG201743} on $N=gnu+s$ qubits which have logical codewords
\begin{align}
|j_{g,n,u,s}\> &= 2^{-(n-1)/2} \sum_{\substack{ {\rm mod}(k,2)=j \\ 0 \le k \le n}} {\binom n k}^{-1/2}  |D^{gnu+s}_{gk+s}\>,
\end{align}
for $j=0,1$, and are attractive candidates for QEC-enhanced sensing \cite{ouyang2022finite}. 
Intuitively, $g$ corresponds to the bit-flip distance, $n$ corresponds to the phase-flip distance. Here, $u\ge 1$ and $s\ge 0$ both specify the number of qubits used, and the distribution of the Dicke weights.
These $s$-shifted gnu codes have distance ${\rm min}\{g,n\}$, which allows the correction of ${\rm min}\{g,n\}-1$ deletion errors,
and the correction of
$\lfloor ({\rm min}\{g,n\}-1)/2 \rfloor $ general errors. $k=0,\dots,n$. 
We depict a gnu code initialised as a logical plus state in Figure \ref{fig:gnucode}.

\begin{figure}[h]
\centering  
\includegraphics[width=0.5\textwidth]{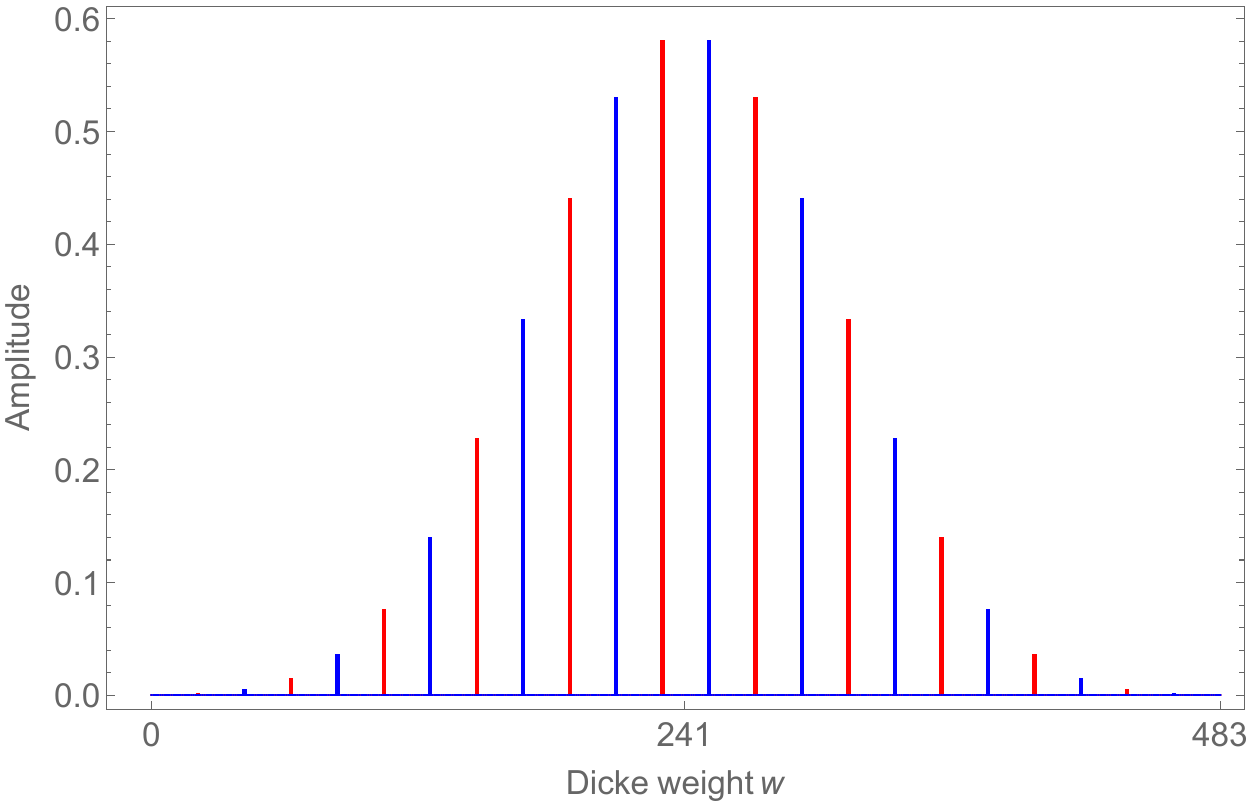} 
 \caption{{\bf A gnu state.} Illustration of a gnu state $|+_{g,n,u,s}\>
 =(|0_{g,n,u,s}\> + |1_{g,n,u,s}\>)/\sqrt 2
 $ where
 $g=21$, $n=2\lfloor g/2\rfloor+1$, $u=1+1/n$, and $s = g$ so that the number of qubits is $N=gn+2g=483$. Here, the horizontal axis depicts the weights $w$ of the Dicke states, and the vertical axis depicts the values of amplitude $\<D^{gnu+s}_w|+_{g,n,u,s}\>$. The colors red and blue correspond to the Dicke state amplitudes of $|0_{g,n,u,s}\>$ and $|1_{g,n,u,s}\>$ respectively, and which are related by a global bit flip. 
}\label{fig:gnucode}
 \end{figure}

%%%%%%%%%%%%%%%%%%%%%%%%%%%%%%%%%%%%%%
%%%%%%%%%%%%%%%%%%%%%%%%%%%%%%%%%%%%%% 
\iftoggle{science}
{\subsection{Correcting \texorpdfstring{$t$}{t} single qubit errors}  
}
{\section{Correcting \texorpdfstring{$t$}{t} single qubit errors}  
}
%%%%%%%%%%%%%%%%%%%%%%%%%%%%%%%%%%%%%%
%%%%%%%%%%%%%%%%%%%%%%%%%%%%%%%%%%%%%% 
 \label{sec:QEC}

We sketch how quantum error correction on arbitrary PI codes can proceed in 
two broad steps.
For general errors, where 
up to $t$ errors afflict the original 
PI code, but no qubits are lost or inserted, we perform the following steps.
First, we measure the total angular momentum on nested subsets of qubits. 
These measurements of total angular momentum occur in the sequentially coupled basis \cite{jordan2009permutational,havlivcek2018quantum}, where 
subsets of qubits that we measure on are 
$\smash{[k] = \{ 1 ,\dots, k \}}$, where $k=1,\dots,N$. The corresponding total angular momentum operators to be measured are
$\hat{J}^2_{[1]}, \dots, \hat{J}^2_{[N]}$ \footnote{Note that throughout we have used the simplified notation for the total angular momentum of all spins as $\hat{J}$ instead of $\hat{J}_{[N]}$.}. where 
\begin{align}
\hat{J}_{[k]}^2 &= {({\hat J}^x_{[k]})}^2 + {(\hat{J}^y_{[k]})}^2 + {(\hat{J}^z_{[k]})}^2,
\end{align}
and
\begin{align}
\hat{J}^x_{[k]} &= \frac{1}{2} \sum_{i =1}^k X_i,\quad
\hat{J}^y_{[k]} = \frac{1}{2}\sum_{i =1}^k Y_i,\quad
\hat{J}^z_{[k]} = \frac{1}{2}\sum_{i =1}^k Z_i.
\end{align}
The eigenvalues of the operators $\hat{J}_{[k]}^2$ are of the form $j_k(j_k+1)$ where $2j_k$ are positive integers. 
After measurement, $\hat{J}_k^2$ gives an eigenvalue of $j_k(j_k+1)$, and we can infer the total angular momentum number $j_k$. These total angular momentum numbers belong to the set
\begin{align}
T_k = \{  k/2 -j : j = 0,\dots, \lfloor k/2 \rfloor, k/2 - j \ge 0 \}.
\end{align}
Since the total angular momentum operators $\hat{J}_{[k]}^2$ all commute, the order of measuring these operators does not affect the measurement outcomes. Hence we may measure $\hat{J}_{[k]}^2$ sequentially; that is we measure $\hat{J}_{[2]}^2$, followed by $\hat{J}_{[3]}^2$, and so on (noting that $\hat{J}_{[1]}^2=\frac{3}{4}$ always). Using the observed total angular momentum $j_1, \dots, j_N$,
we construct standard Young tableaus (SYTs) encapsulating the measurement information.

Second, we bring the state back to the original PI code. There are two approaches for this. 
The first approach uses the inverse of a quantum Schur transform \cite{BCH-PRL-2006-schur,kirby2018practical,havlivcek2018quantum,krovi2019efficient,havlivcek2019classical,pearce2022multigraph} to bring the state back into the computational basis, 
after which geometric phase gates \cite{PhysRevA.110.062610} bring the state back into the original PI codespace.

The second approach uses a
teleportation-based approach based on geometric phase gates.
Here, we prepare an ancilla PI logical state in register A, 
and have the state after the total angular momentum measurements in register B. The state in register B is essentially in a codespace labelled by an SYT. 
Then we perform logical CNOT gate with control on register A and target on register B, before performing a logical Z measurement on register B. 
Depending on the outcome of this measurement, we apply a logical X correction on register A.  
The teleportation procedure is agnostic of which SYT the state in register B was in, and how many qubits are in register A and B. Such a logical CNOT can be performed using geometric phase gates \cite{ouyang2022finite}.

% \subsection{}
\label{ssec:correcting-t-errors}
Suppose that a noisy quantum channel $\mathcal N$ has Kraus operators that act non-trivially on at most $t$ qubits.
Then any distance $d$ PI code can correct errors introduced by $\mathcal N$, provided that $d \ge 2t+1.$
Here, we outline a two-stage quantum error correction procedure to correct these errors.
Stage 1 projects the corrupted state onto the irreducible representations of the symmetric group $S_N$ by measuring the total angular momentum of subsets of qubits.
Stage 2 does further projections within the irreducible representations and finally performs a unitary that brings the state back to the codespace. 

In Section \ref{sssec:YDYT}, after reviewing standard Young tableau, we describe the implementation of Stage 1 by measuring the total angular momentum of subsets of qubits as described in Section \ref{sssec:syndrome-extraction}.

%%%%%%%%%%%%%%%%%%%%%%%%%%%%%%%
\iftoggle{science}
{}
{\subsection{Young diagrams and Young tableau}
\label{sssec:YDYT}
%%%%%%%%%%%%%%%%%%%%%%%%%%%%%%%
Consider Young diagrams \cite[Page 29]{stanley-enumerative-v1} comprising $N$ boxes arranged in two left-justified rows.
We restrict our attention to Young diagrams with $r_1$ boxes on the first row and $r_2$ boxes on the second row where $r_1 \ge r_2$ and $r_1 + r_2 = N$. 
\begin{align}
\includegraphics[width=0.5\textwidth]{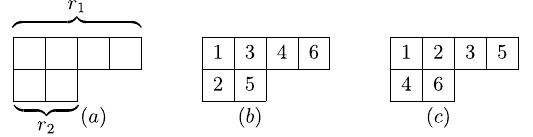} 
\label{eq:YT}
\end{align}
In \eqref{eq:YT}(a), we depict a Young diagram with four boxes on the first row, and two boxes on the second row. 
Each Young diagram corresponds to an integer partition of $N$ with two parts.
A standard Young Tableau (SYT) is obtained by filling up the $N$ boxes in a Young diagram with integers from 1 to $N$ such that the integers strictly increase from left to right, and from top to bottom.
Given a Young diagram, for instance in \eqref{eq:YT}(a), we give two examples of SYTs that can be obtained in \eqref{eq:YT}(b) and \eqref{eq:YT}(c) respectively.
We can enumerate the number of SYTs consistent with a given Young diagram using the hook-length formula \cite[Corollary 7.21.6]{stanley-enumerative-v2}.
The hook-length formula states that the number of SYTs consistent with any Young diagram with $N$ boxes is equal to $N!$ divided by the hook-length of each box. The hook-length of a given box is the total number of boxes in its hook, where the box's hook includes the box itself and all other boxes to its right and bottom. 
\begin{align}
\includegraphics[width=0.5\textwidth]{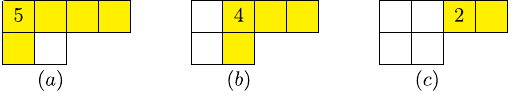} 
\label{eq:hooklengths}
\end{align}
In \eqref{eq:hooklengths}, we shade the hooks of the labelled cells. 
In \eqref{eq:hooklengths}(a), (b) and (c), the hook-lengths are five, four and two respectively.
From the hook-length formula, the number of SYTs corresponding to a Young diagram with two rows is
\begin{align}
&\frac{N!}{(N-r_1)! (r_1+1)!/ (r_1-N+r_1+1)}\notag\\
=& \binom N{r_1} \frac{2r_1 - N + 1}{r_1+1}.\label{eq:num_SYT}
\end{align}

We can also represent SYTs as a binary vector, using the so-called Young-Yamanouchi basis. 
The binary vector $(x_1,\dots, x_N)$
is such that $x_j=1$ if the symbol $j$ in the SYT is in the second row, and $x_j=0$ if the symbol $j$ in the SYT is in the first row. Operationally, $x_j=0$ if the measured total angular momentum increases, and $x_j=1$ if the measured total angular momentum decreases.

Apart from SYTs, we also consider semistandard Young tableau (SSYT), where boxes are filled with integers that strictly increase from top to bottom and are non-decreasing from left to right. 
Here, we restrict our attention to SSYTs filled with the numbers 1 and 2. 
The number of such SSYTs obtainable from a Young diagram with two rows is 
\begin{align}
r_1-r_2+1.\label{eq:num_SSYT}
\end{align} 
We list all SSYTs filled with the numbers 1 and 2 when $r_1 = 4$ and $r_2=2$ in \eqref{eq:SSYTs} below.
%%%%%%%%%%%%%
 \begin{align}
\includegraphics[width=0.5\textwidth]{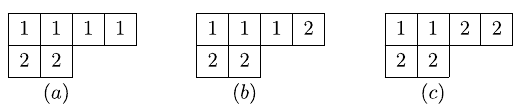} 
\label{eq:SSYTs}
 \end{align}
%%%%%%%%%%%%% 
From Schur-Weyl duality \cite[Chapter 9.1]{goodman2000representations} applied to quantum information theory \cite{BCH-PRL-2006-schur,harrow2013church}, we know that 
the $N$-qubit space $(\mathbb C^2)^{\otimes N}$ is isomorphic to 
\begin{align}
\bigoplus_{D} \mathcal Q^D \otimes \mathcal P^D, \label{eq:schurweyl}
\end{align}
where each $D$ denotes a Young diagram with $N$ boxes and two rows.
For every Young diagram $D$, $\mathcal Q^D$ is a space with basis elements labelled by all possible SYTs filled with integers 1 to $N$, and $\mathcal P^D$ is a space with basis elements labelled by all possible SSYTs filled with integers 1 and 2.
Since the basis of $\mathcal Q^D$ is labeled by SYTs of shape $D$, we can write $\mathcal Q^D = {\rm span}\{|\textsf{T}\> : \textsf{T} \mbox{ is a SYT for } D\}$.

Since $D$ always has two rows, we can represent it with the tuple $(r_1, r_2)$ where $r_1$ and $r_2$ count the number of boxes in the first and second rows respectively.
When $D=(N,0)$, we can see from \eqref{eq:num_SYT} and \eqref{eq:num_SSYT} that the dimension of $\mathcal Q^{(N,0)}$ is one, and the dimension of $\mathcal P^{(N,0)}$ is $N+1$. In fact, $\mathcal P^{(N,0)}$ corresponds precisely to the symmetric subspace of an $N$-qubit symmetric state.

\subsection{Symmetrizing lemma}

When a channel acts identically and independently on every qubit, it 
always maps a pure symmetric state to a density matrix that is block diagonal on the spaces $\mathcal Q^D \otimes \mathcal P^D$ as given \eqref{eq:schurweyl}.
For such block diagonal states, QEC can proceed by projecting the density matrix onto one of the blocks labelled by $D$.

However the density matrices that we encounter in the QEC of PI codes do not necessarily have this block diagonal structure.
In this scenario, such density matrices can always be made block diagonal in the Schur-Weyl basis by applying symmetrizing operations described by the quantum channel $\mathcal S$ that randomly permutes qubits, and has Kraus operators $\{\frac{1}{\sqrt N!} P_{\sigma} : \sigma \in S_N \}$,
where $P_\sigma$ denotes an $N$-qubit matrix representation of a permutation operator that permutes qubit labels according to the permutation $\sigma$.
In the following lemma, we prove that if we apply the symmetrizing channel, originally correctible errors remain correctible. 
%%%%%%%%%%%%%%%%%%
\begin{lemma}[Symmetrizing lemma]
\label{lem:sym}
Let $\mathcal C$ be any $N$-qubit PI code of distance $d$. Let $\mathcal N$ be any quantum channel with Kraus operators $K$ of weight at most $t$.
Then if $d \ge 2t+1$, the channels $\mathcal N$ and $\mathcal S \circ \mathcal N$ are both correctible with respect to $\mathcal C$.
\end{lemma}
%%%%%%%%%%%%%%%%%%
%%%%%%%%%%%%%%%%%%
\begin{proof}
The channel $\mathcal N$ is correctible, and hence satisfies the fundamental quantum error correction criterion \cite{KnL97}. To prove the lemma, we must show that the channel $\mathcal S \circ \mathcal N$ is also correctible.

%Now let $K$ denote a set of Kraus operators for the quantum channel $\mathcal N$.
Now, denote $\bar K = \{\frac{1}{\sqrt N!} P_{\sigma} A : \sigma \in S_N , A \in K \}$ as a set of Kraus operators for the quantum channel $\mathcal S \circ \mathcal N$. 
The Kraus operators in $\bar K$ are correctible if and only if for every $\sigma, \tau \in S_N$ and $A,B \in K$, there exists a $g_{A,B,\sigma,\tau} \in \mathbb C$ such that 
\begin{align}
\Pi  A^\dagger P_{\sigma}^\dagger  P_\tau B  \Pi
= g_{A,B,\sigma, \tau} \Pi ,
\label{qecc_of_SN}
\end{align}
where $\Pi$ the code projector for $\mathcal C$.

Since $\Pi$ is a projector onto the symmetric subspace, for all $\sigma \in S_n$,
we have that $P_\sigma \Pi = \Pi P_\sigma = \Pi.$
Denoting $A_\sigma = P_{\sigma} A P_{\sigma}^\dagger $
and $B_\sigma = P_{\tau} B P_{\tau}^\dagger$, note that \eqref{qecc_of_SN} is equivalent to 
\begin{align}
\Pi  A_\sigma^\dagger  B_\tau  \Pi = g_{A,B,\sigma, \tau}\Pi .
\label{qecc_of_SN2}
\end{align}
Since $A,B$ are operators with weight at most $t$, 
$A_\sigma$ and $B_\tau$ must also be operators of weight at most $t$. 
Hence, both $A_\sigma$ and $B_\tau$ are linear combinations of Pauli operators of weight at most $t$. 
Namely, 
\begin{align}
A_\sigma = \sum_{P : |P|\le t } a_{A,P,\sigma} P,
\quad
B_\tau     = \sum_{P : |Q|\le t } b_{B,Q,\sigma} P,
\end{align}
where $a_P$ and $a_Q$ are real coefficients. 
From this, it follows that the left side of \eqref{qecc_of_SN2} is equivalent to 
\begin{align}
\Pi  A_\sigma^\dagger  B_\tau  \Pi  = \sum_{P,Q} a_{A,P,\sigma}^* b_{B,Q,\sigma}  \Pi  P Q \Pi.
\label{qecc_of_SN3}
\end{align}
From the Knill-Laflamme condition \cite{KnL97},
since $\mathcal C$ is a code of distance $d$, for every Pauli $P$ and $Q$ of weight at most $t$, there exists a $c_{P,Q} \in \mathbb C$ such that 
\begin{align}
\Pi  P Q \Pi = c_{P,Q} \Pi.
\label{distance-KL}
\end{align}
Substituting \eqref{distance-KL} into the left side of \eqref{qecc_of_SN3},
we can conclude that 
\begin{align}
\Pi  A_\sigma^\dagger  B_\tau  \Pi 
= \sum_{P,Q} a_{A,P,\sigma}^* b_{B,Q,\sigma}  c_{P,Q}\Pi
\label{qecc_of_SN4}.
\end{align}
This implies that 
$\Pi  A^\dagger P_{\sigma}^\dagger  P_\tau B  \Pi
= g_{A,B,\sigma, \tau} \Pi ,$
where $
% \begin{align}
g_{A,B,\sigma, \tau} = \sum_{P,Q} a_{A,P,\sigma}^* b_{B,Q,\sigma}  c_{P,Q},
% \end{align}
$
and this proves that $\mathcal S \circ \mathcal N$ is also correctible with respect to the code $\mathcal C$.
\end{proof}
Lemma \ref{lem:sym} hints at how we can perform quantum error correction on any PI quantum code.
Namely, if any correctible channel $\mathcal N$ introduces errors on a PI quantum code, 
we can project the state into the Schur-Weyl basis and still be able to correct the resultant errors. 
This is because a symmetrizing channel $\mathcal S$ makes a quantum state block-diagonal in the Schur-Weyl basis, and Lemma \ref{lem:sym} tells us that if $\mathcal N$ is correctible, $\mathcal S \circ \mathcal N$ is also correctible. 
Note that we do not have to actually physically symmetrize the channel, if we randomize the labels in our minds.
Section \ref{sssec:syndrome-extraction} illustrates how we may project the state onto the diagonal blocks in the Schur-Weyl basis.
%%%%%%%%%%%%%%%%%%
}

%%%%%%%%%%%%%%%%%%%%%%%%%% 
\iftoggle{science}
{\subsubsection{Syndrome extraction by measuring total angular momentum}
}
{\subsection{Syndrome extraction by measuring total angular momentum}
}
\label{sssec:syndrome-extraction}
%%%%%%%%%%%%%%%%%%%%%%%%%% 
We expand on the idea of 
syndrome extraction by measuring total angular momentum \cite{ouyang2022finite}.
For this, we take as input a quantum state $\rho$, 
uses the total angular momenta $j_1 , \dots, j_N$ of qubits in $[1] ,[2], \dots, [N]$ to iteratively construct a 
\iftoggle{science}
{standard Young Tableau (SYT) $\textsf{T}$ (which we explain in the Materials and Methods section) by adding one labelled box at a time.}
{SYT $\textsf{T}$ by adding one labelled box at a time.}
This outputs the resultant state $\rho_\textsf{T}$, the SYT $\textsf{T}$ and $j_\textsf{T} = j_N$.

% \newline
% \begin{ytableau}
%   \none[]  & 1  & 3  & 4  & 6  \\
%   \none[]  &  2 &  5 & \none &  \none \\
%   \none & \none[] & \none[(b)] & \none[]\\
% \end{ytableau}
%
%\noindent {\bf Algorithm} \texttt{SyndromeSYT}
%\begin{enumerate}
%\item Set $i = 2$,  and set $Y$ to be a SYT with one box on the first row filled with the number 1, and no boxes on the second row.
%\item Set the current total angular momentum $j = 1/2.$  
%\item For $ k = 2:N$
%\item \quad If $j_k  = j+1/2$, update $Y$ by creating a box labelled by $k$ to the right side of the first row. Set $ j \gets j + 1/2$.
%\item \quad If $j_k  = j-1/2$, update $Y$ by creating a box labelled by $k$ to the right side of the second row. Set $ j \gets j - 1/2$.
%\item Endfor.
%\end{enumerate} 

Here, $j_\textsf{T}$ denotes the total angular momentum of the $N$-qubits associated with the SYT $\textsf{T}$. Combinatorially, for a SYT with $r_1$ and $r_2$ boxes on the first and second rows respectively, $j_\textsf{T} = (r_1-r_2)/2$.
Figure \ref{fig:syt-c} depicts how we obtain an SYT syndrome by measuring total angular momentum on consecutive nested subsets of qubits.

\begin{figure}[h]
\centering  
\includegraphics[width=0.5\textwidth]{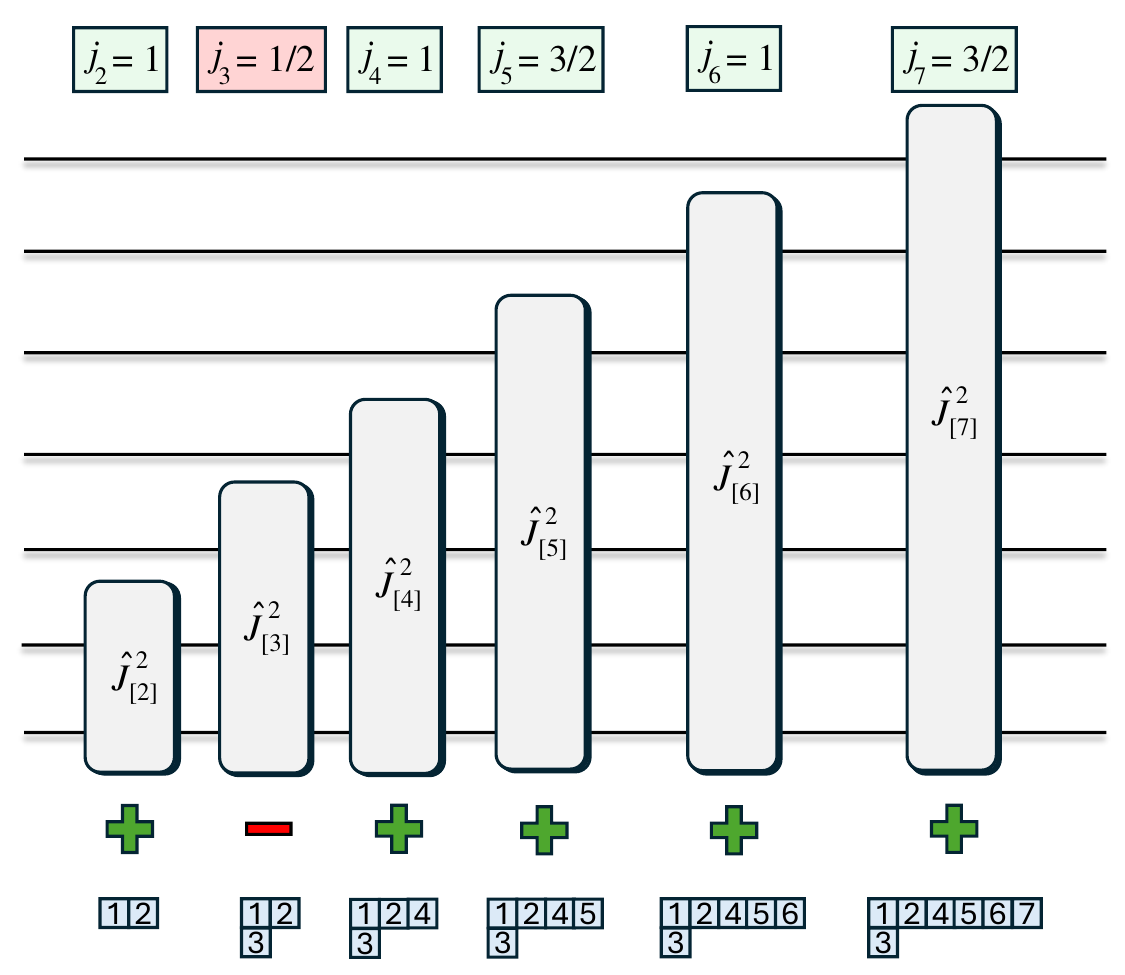} 
 \caption{{\bf Measurement of the total angular momentum of nested consecutive sets of qubits.} This measurement collapses the state onto an SYT. 
 When total spin $j_k$ increases, we add a numbered box to the first row. Otherwise, $j_k$ decreases, and we add a numbered box to the second row. The `plus' and `minus' signs depicted when mapped to `0' and `1' respectively, give a binary string that is the Young-Yamanouchi representation of the SYT.  
}\label{fig:syt-c}
 \end{figure}

After obtaining a SYT $\textsf{T}$, we know that $\rho_T$ must have support within a subspace of $\mathcal P^\textsf{T}  = |\textsf{T}\> \otimes \mathcal P^D. $
We can enumerate the number of SSYTs to determine $|\mathcal P^\textsf{T}|$; from \eqref{eq:num_SSYT}, we find that
\begin{align}
|\mathcal P^\textsf{T}|  = (N/2+j_\textsf{T})-(N/2-j_\textsf{T})+1 = 2j_\textsf{T}+1.
\end{align}
This is sensible since $\mathcal P^\textsf{T}$ corresponds to a space with total angular momentum $j_\textsf{T}$ on all $N$ qubits. 
Using information about the SYT $\textsf{T}$, we can perform quantum error correction on the spaces $\mathcal P^\textsf{T}$.
Denoting $\Pi^{\textsf{T}}$ as the projector onto $\mathcal P^\textsf{T}$,
    we point out that $\Pi^{\textsf{T}}$ commutes with the operator ${\hat J^z}$, because ${\hat J^z}$ commutes with the $\hat{J}^2_{[k]}$. Hence, the eigenvectors of ${\hat J^z}$ are also eigenvectors of $\Pi^{\textsf T}$.
In fact, we can write the projector $\Pi^{\textsf T}$ as
\begin{align}
\Pi^\textsf{T} =  \sum_{m=-j_\textsf{T}}^{j_\textsf{T}} |m_\textsf{T}\>\<m_\textsf{T}|,
\end{align}
where $\{|m_\textsf{T}\> :   m = -j_\textsf{T} ,  \dots, j_\textsf{T} \}$ denotes an orthonormal basis of $\mathcal P^\textsf{T}$ and 
\begin{align}
\hat{J}^z |m_\textsf{T}\> = m|m_\textsf{T}\>.
\end{align}
Here, the states $|m_\textsf{T}\>$ are magnetic eigenstates of the operator $\hat{J}^z$.
Note that when ${\textsf{T}}$ is the SYT with only one row, $|m_\textsf{T}\> = |D^N_{m+j_\textsf{T}}\>$ and $|m_\textsf{T}\>$ is a Dicke state.

For any positive integer $g$, we can decompose the projector $\Pi^{\textsf{T}}$ as 
\begin{align}
    \Pi^{\textsf{T}}  = 
    \sum_ {a = 0, \dots, g-1 }
    \Pi^{\textsf{T}}
    _{{\rm mod},g,a}
\end{align}
where 
\begin{align}
\Pi^{\textsf{T}}_{{\rm mod},g,a} = \!\!\!\!\!\!\!\!\!\!\!\!\!
    \sum_{\substack{
    k \in \mathbb Z \\ 
-j_\textsf{T} \le  gk + a -j_\textsf{T}   \le j_\textsf{T} \\
    }} 
    \!\!\!\!\!\!\!\!\!\!
    |(gk+a-j_\textsf{T} )_\textsf{T}\rangle \langle(gk+a-j_\textsf{T} )_\textsf{T}| .
\end{align}
Here, we can see that the projector 
$\Pi^{\textsf{T}}_{{\rm mod},g,a}$ projects onto the span of states $|m_\textsf{T}\>$ where ${\rm mod}(m+j_\textsf{T},g)=a$, that is, the space where magnetic quantum numbers plus $j_\textsf{T}$ are equivalent to $a$ modulo $g$ for a particular SYT \textsf{T}.

By defining the projector 
\begin{align}
    Q_{g,a} =  \sum_{{\rm SYTs} \ \textsf{T}}    
   \Pi^{\textsf{T}}_{{\rm mod},g,a},
\end{align}
note that for any positive integer $g$, we have 
\begin{align}
    \sum_{a=0}^{g-1} Q_{g,a} = I_N. 
\end{align}
The projector 
$ Q_{g,a}$ projects onto the span of states $|m_\textsf{T}\>$ where ${\rm mod}(m+j_\textsf{T},g)=a$ for any SYT \textsf{T}.

Given a set of projectors 
$P = \{P_1,\dots, P_k\}$ such that $\sum_j P_j = I_N$ and an $N$-qubit density matrix to measure, 
let us denote the output state $\rho'$ of a projective measurement on a density matrix $\rho$ with respect to $P$ using the notation 
\begin{align}
\rho' = \texttt{ProjMeas}(\rho,P).
\end{align}
Next we introduce 
\texttt{ModuloMeas}, which for given fixed $g$, performs a projective measurement on quantum state $\rho$ according to the set of projectors $\{ Q_{g,0}, \dots, Q_{g,g-1} \}$.

\iftoggle{science}
{\subsubsection{Amplitude rebalancing}}
{\subsection{Amplitude rebalancing}}

Amplitude rebalancing is an adaptive multistep algorithm wherein each step depends on the output of all previous steps. 
As input, the algorithm takes a state promised to be in a codespace $\mathcal C$ of a known distance $d$ quantum code that encodes a single logical qubit.
We can write any such  input state as 
\begin{align}
     |\psi\> = 
     \cos \theta |0_L\> + e^{i \phi }\sin \theta |1_L\>\label{w-input},
 \end{align}
 where $\theta, \phi \in \mathbb R$ and 
 $\{|0_L\>,|1_L\>\}$ are orthonormal vectors that span the codespace $\mathcal C$. 
 The objective of the amplitude rebalancing algorithm is to map $|\psi\>$ to a state
\begin{align}
     |\psi_w\> = 
     \frac{ \sqrt{3+w}  \cos \theta |0_L\> 
     + e^{i \phi }\sqrt{3-w} \sin \theta |1_L\>}
     {\sqrt{ 3 + w \cos 2\theta } }, \label{w-form}
 \end{align}
 for a given real $w \in [1/2,-1/2]$.
 Intuitively, the effect of the amplitude rebalancing algorithm deforms the scalings of the logical codewords by a real amount.

The $k$th step of the algorithm performs a projective measurement that projects the state onto one of two orthogonal spaces 
$\mathcal S_{k,1}, S_{k,2}$.
Information about the measurement outcomes of the preceding $k-1$ steps is stored as a binary vector ${\bf x^{(k)}}\in \{0,1\}^{k-1}$.
The exact form for the spaces
$\mathcal S_{k,1}, S_{k,2}$
depends on the value of ${\bf x^{(k)}}$.
Of these spaces $\mathcal S_{k,1},  \mathcal S_{k,2}$, the space
$\mathcal S_{k,1}$ is the preferred space that we project onto, 
and the projection probability on space $\mathcal S_{k,1}$ occurs with probability strictly greater than 1/2.
The multiple time steps allow us to, with high probability, shift the relative amplitudes of the logical zero and logical one operators to our desired value.

Suppose that the state \eqref{w-input} is the input state for step $k$. Then projection space $\mathcal S_{k,1}$, and a subsequent mapping back to the codespace, gives the resultant state of the form \eqref{w-form}. The parameter $w$ in \eqref{w-form} is chosen according to the binary vector ${\bf x}^{(k)}$.

Next, we proceed to describe the space $\mathcal S_{k,1}$ and $\mathcal S_{k,2}$. For convenience, we will 
denote the spaces $\mathcal S_{k,1}$ and $\mathcal S_{k,2}$ as 
$\mathcal S_{1}$ and $\mathcal S_{2}$.
Suppose that $E$ is an operator such that 
\begin{align}
    \<j_L|E|k_L\> = 0, \label{KL1}
\end{align}
for all distinct $j,k$ and where 
\begin{align}
\<j_L|E|j_L\>  = 
\<k_L|E|k_L\>\label{KL2}
\end{align}
for all $j,k$.
There exists such a non-trivial operator for any PI code of distance at least two.
An example of an operator $E$ for PI codes is the angular momentum operator $\hat J_z$.

Using the Gram-Schmidt orthonormalization procedure on 
the vectors $|0_L\>$ and $E|0_L\>$,
we can derive orthonormal vectors 
$|0_L\> $ and $|0'_L\>$, and similarly orthonormal vectors
$|1_L\> $ and $|1'_L\>$.

Now we define the states
\begin{align}
|0_w\>  &\coloneqq 
(\sqrt{3+w} |0_L\> + \sqrt{1-w}|0'_L\>)/2,\\
|1_w\>  &\coloneqq 
(\sqrt{3-w} |1_L\> + \sqrt{1+w}|1'_L\>)/2.
\end{align}
We also define the states
\begin{align}
|\bar 0_w\>  &\coloneqq 
(\sqrt{1-w} |0_L\> - \sqrt{3+w}|0'_L\>)/2,\\
|\bar 1_w\>  &\coloneqq 
(\sqrt{1+w} |1_L\> - \sqrt{3-w}|1'_L\>)/2.
\end{align}
Clearly, the states 
$|0_w\>, |1_w\>, |\bar 0_w\>, |\bar 1_w\>$ are pairwise orthonormal.

We define the projector associated with $\mathcal S_1$
 as 
\begin{align}
\Pi_{w} = |0_w\>\<0_w| +   |1_w\>\<1_w|,
\end{align}
and the projector associated with
$\mathcal S_2$ as
\begin{align}
\bar \Pi_{w} = |\bar 0_w\>\<\bar 0_w| +   |\bar 1_w\>\<\bar 1_w|.
\end{align}
We can calculate the projection of $|\psi\>$ onto $\mathcal S_1$ as
\begin{align}
\Pi_w |\psi \> =
\frac{\sqrt{3+w }}{\sqrt 6} \cos \theta |0_w\>
+
\frac{\sqrt{3-w}}{\sqrt 6} e^{i\phi }\sin \theta |1_w\>
\end{align}
and the projection of 
$|\psi\>$ onto $\mathcal S_2$ as
\begin{align}
\Pi_w |\psi \> =
\frac{\sqrt{1-w }}{2} \cos \theta |0_w\>
+
\frac{\sqrt{1+w}}{2} e^{i\phi }\sin \theta |1_w\>.
\end{align}
 The probabilities of projecting onto the spaces 
 $\mathcal S_1$
 and 
 $\mathcal S_2$
 are given respectively as 
\begin{align}
\frac{3}{4} + \frac{w}{4} \cos 2\theta  ,
\quad
\frac{1}{4} - \frac{w}{4} \cos 2\theta  .
\end{align}
Since $w\in [1/2,1/2]$,
the probability of projecting onto $\mathcal S_1$ is at least 5/8.
  
Therefore, in each step, we project the state onto one of two two-dimensional spaces $\mathcal S_1$ and $\mathcal S_2$.    We interpret the projection of the state onto $S_1$ as a successful amplitude rebalancing step. In this case, we append a `0' to the binary vector ${\bf x}^{(k)}$.
If we fail to project the state onto $S_1$, the state is in a space $S_2$, and that is an unsuccessful amplitude rebalancing step, and we append a `0' to the binary vector ${\bf x}^{(k)}$.
After projections onto either $S_1$ or $S_2$,
we apply a conditional unitary that maps the state back to the codespace $\mathcal C$.
As we can see, the binary vector 
${\bf x}^{(k)}$ keeps track of the relative shifts in the amplitudes of the logical zero and logical one states within the codespace.
Within each time step, we select the value of the shift-parameter $w$ according to how we would like to shift the relative amplitudes of the logical zero and logical one in the case of when we are lucky and project onto $\mathcal S_1$.

We can extend this algorithm to PI code that encodes more than one logical qubit.

To see this, note that 
for all $j \ge 2$, we can calculate
\begin{align}
    \Pi_w |j_L\> = 0,
    \quad 
    \bar \Pi_w |j_L\> = 0,
\end{align}
because this arises trivially from the Knill-Laflamme error correction condition \eqref{KL1}, \eqref{KL2} with respect to the operator $E$ as long as the PI code has a distance at least 3.
Therefore, we can adapt the above algorithm to pairs of logical codewords, until all the logical codewords are balanced.

Furthermore, we will converge exponentially fast to a successful amplitude rebalancing projection in the number of repetitions of the projective measurement.

%%%%%%%%% 
\iftoggle{science}
{\subsubsection{Further projections and recovery}}
{\subsection{Further projections and recovery}}
%%%%%%%%%
When a channel's Kraus operators $K_1, \dots, K_a$ are correctible, 
from Lemma \ref{lem:sym}, the operators $\Pi^\textsf{T}K_1, \dots, \Pi^\textsf{T}K_a$ must also be correctible for every SYT $\textsf{T}$.
From the Knill-Laflamme quantum error criterion \cite{KnL97}, we know that $\mathcal P^\textsf{T}$ partitions into orthogonal correctible spaces and a single uncorrectible space.
To describe the correctible subspace within $\mathcal P^\textsf{T}$, 
we use the vectors $\Pi^\textsf{T}K_i|j_L\>$, 
where $|j_L\>$ are logical codewords of the PI code of dimension $M$ and $j=0,\dots, M-1$.

Next, for every $j=0,\dots, M-1$, the Gram-Schmidt process \cite[Section 0.6.4]{horn2012matrix} takes as input the sequence of vectors
\begin{align}
(\Pi^\textsf{T}K_1 | j_L\>,
\dots,
\Pi^\textsf{T}K_a | j_L\>),
\end{align}
and outputs the sequence of orthonormal vectors 
\begin{align}
(|{\bf v}^\textsf{T}_{1,j}\>,
\dots,
| {\bf v}^\textsf{T}_{r_\textsf{T},j}\>)
\label{seq:v-vectors}
\end{align}
that span the same space $\mathcal A^\textsf{T}_j$ of dimension $r_\textsf{T}$. Note that $r_\textsf{T}$ counts the number of correctible spaces.
From the Knill-Laflamme quantum error correction criterion, the spaces $\mathcal A^\textsf{T}_{j}$ and $\mathcal A^\textsf{T}_{k}$ that correspond to 
the $j$th and $k$th logical codewords respectively are pairwise orthogonal for distinct $j$ and $k$ and both have dimension $r_\textsf{T} $\cite{KnL97}.
Hence, by the pigeon-hole principle,  
\begin{align}
r_\textsf{T} \le |\mathcal P^\textsf{T} |/M.
\end{align}
Since $|\mathcal P^\textsf{T}| = 2j_\textsf{T}+1$, the number of these correctible subspaces is at most $(2j_\textsf{T} + 1)/M$, where $M$ is the number of logical states the PI code has.

\begin{figure}[t]
\centering  
\includegraphics[width=0.5\textwidth]{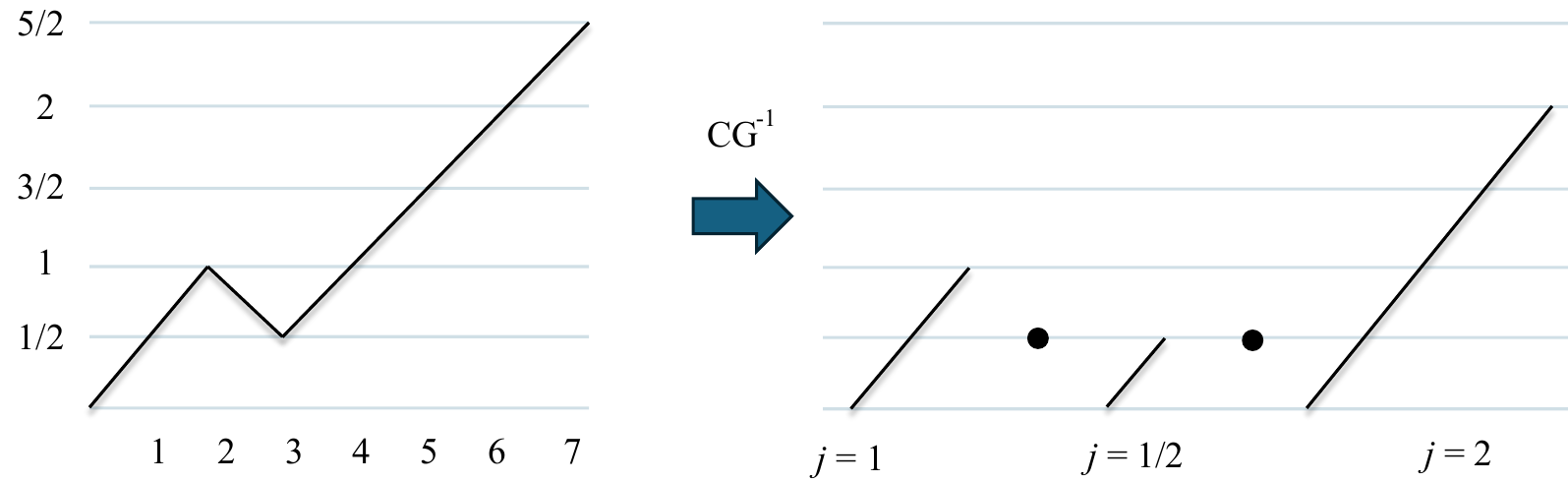} 
 \caption{{\bf Decoupling.} (Left) A Bratelli diagram illustrating how a given SYT $\textsf{T}$, here the one appearing in Fig.~\ref{fig:cgs}, can be visualized as a sequential process of addition of elementary spin$-1/2$ qubits. The kink in the middle represents an event where coupling spin number $3$ to the joint $j=1$ space of the first two reduced the total spin angular momentum to $j=1/2$. (Right) An inverse Clebsch-Gordan transformation decouples the spin 5/2 state into three spin states, one with spin 1, one with spin 1/2 and another with spin 2.
 }\label{fig:syt-bratelli}
 \end{figure}

From the sequence of vectors in \eqref{seq:v-vectors}, for $k=1,\dots, r_\textsf{T}$, we also 
define the spaces $\mathcal C^\textsf{T}_k$ spanned by the vectors 
$|{\bf v}^\textsf{T}_{k,0}\>, \dots, | {\bf v}^\textsf{T}_{k,M-1}\>$, and with corresponding projectors
\begin{align}
    \Pi^{\textsf T}_{k} = \sum_{j=0}^{M-1}
    |{\bf v}^\textsf{T}_{k,j}\>
    \<{\bf v}^\textsf{T}_{k,j}|.
\end{align}
Here, we interpret $k$ as a label on the PI code's correctible subspaces within $\mathcal P^\textsf{T}$.
From the Knill-Laflamme quantum error correction criterion, these spaces $\mathcal C^\textsf{T}_k$ and $\mathcal C^\textsf{T}_{k'}$ 
are pairwise orthogonal for distinct $k$ and $k'$. 

To perform QEC after obtaining the syndrome $\textsf{T}$, it suffices to execute the following steps that arise from the Knill-Laflamme QEC procedure \cite{KnL97}. Abstractly, the steps are as follows. First, we perform a projective measurement according to the projectors $\Pi^{\textsf T}_{k}$ for $k=1,\dots, r_\textsf{T}$. If the obtained state is not in $\mathcal C^\textsf{T}_k$ for any $k$, then we have an uncorrectible error. Otherwise, the obtained state is in $\mathcal C^\textsf{T}_{k}$ for some $k = 1,\dots, r_\textsf{T}$. Finally, if the error is correctable, we perform a unitary operator that maps 
$    |{\bf v}^\textsf{T}_{k,j}\>$ to $|j_L\>$ for every $j = 0, \dots, M-1$. This completes the QEC procedure.
   
\begin{figure*}[t]
\centering  
\includegraphics[width=\textwidth]{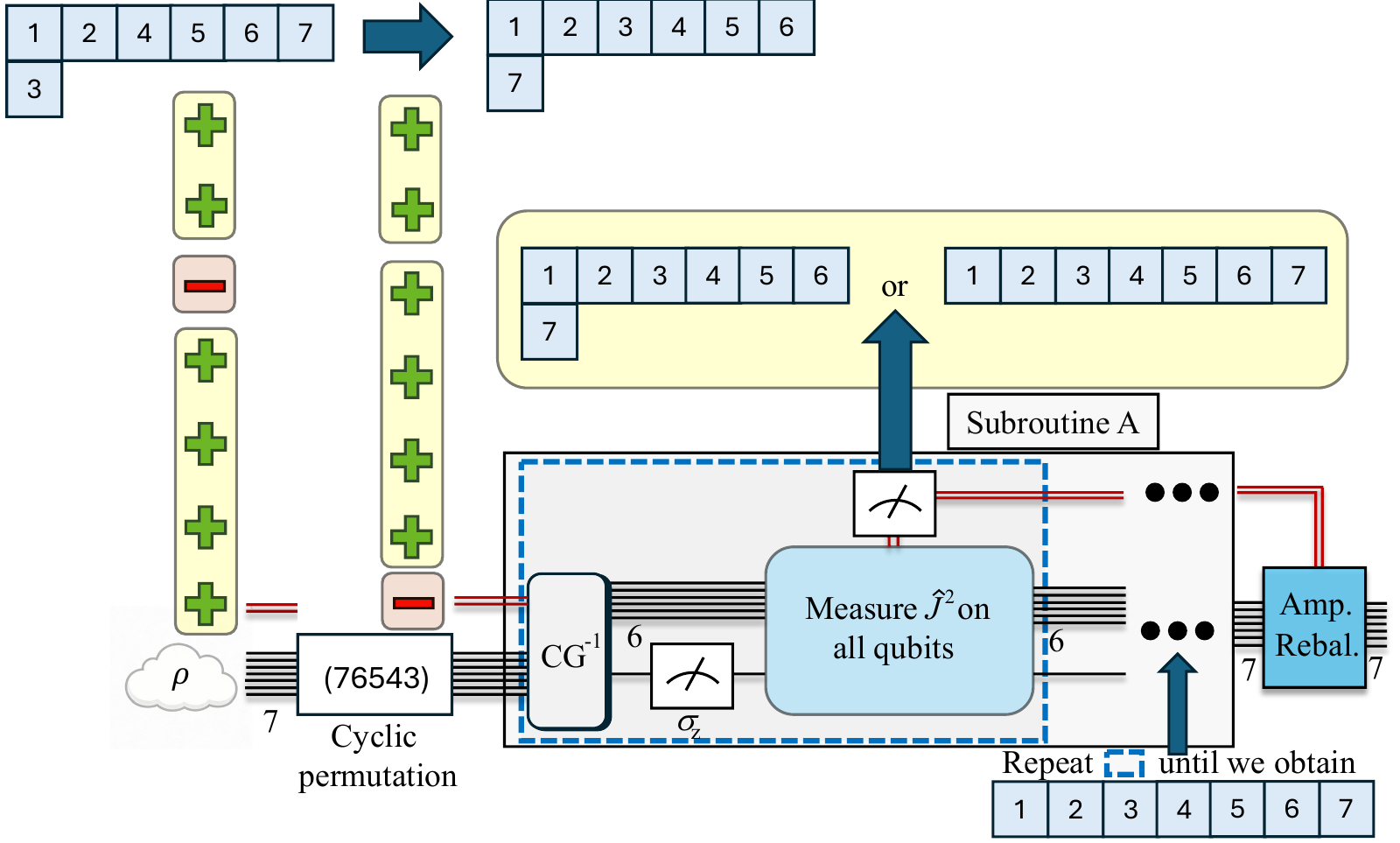} 
 \caption{{\bf QEC after obtaining the SYT syndrome.}
 An error occurs on a seven-qubit PI code \cite{PoR04,kubischta2023not}, and a SYT syndrome measurement obtains the SYT depicted on the upper left. This algorithm proceeds using this SYT syndrome information. 
 The classical input of the left side of the circuit depicts the Young-Yamanouchi representation ${\bf x}$ of this SYT. We permute the qubits so that the qubit that correspond to the component in ${\bf x}$ equal to 1 is rearranged as the last qubit.
 We apply Subroutine A, wherein we apply an inverse Clebsch-Gordan transformations followed by a computational basis measurement of the decoupled qubit and subsequently, a measurement of total angular momentum. We repeat Subroutine A until we obtain a projection onto maximal angular momentum space. 
 Each pair of red lines depicts one bit of classical information, and each black depicts one qubit of quantum information.
 After successful projection onto the symmetric space, we perform the amplitude rebalancing algorithm to rebalance the deformation of the amplitudes of the logical codewords that happened during the course of performing Subroutine A.}\label{fig:cgs}
 \end{figure*}

In contrast to this abstract procedure, we can also describe a more explicit implementation of an equivalent QEC procedure after obtaining the syndrome $\textsf{T}$.
For a PI code with logical codewords
\begin{align}
|j_L\> = \sum_{k} a_k |D^N_k\> 
\end{align}
for $j=0,\dots, M-1$,
we define corresponding $\textsf{T}$-codes to have (subnormalized) logical codewords
\begin{align}
|{\rm code}_{j,\textsf{T}}\>
    &=
\sum_{k} a_k |(k/2-j_\textsf{T})_\textsf{T}\>.
\end{align}
When the support of the PI code in the Dicke basis is appropriately restricted and when the SYT syndrome $\textsf{T}$ is correctible, 
the $\textsf{T}$-code's logical codewords can be normalized.
For simplicity, we continue the analysis assuming that $|{\rm code}_{j,\textsf{T}}\>$ are normalized vectors for correctible SYT syndromes $\textsf{T}$.

We can obtain a \textsf{T}-code using the following operations.\newline

\noindent{\bf Knill-Laflamme Recovery}
\begin{enumerate}
\item Apply a unitary 
$W_{\textsf{T}}$, where for every $k = 1,\dots, r_\textsf{T}$ and $j=0,\dots, M-1$, we have
\begin{align}
W_{\textsf{T}} :
|{\bf v}^\textsf{T}_{k,j}\>
\to
|{ j r_\textsf{T} + (k-1)   -j_\textsf{T} })
_{ \textsf{T}}\>.
\end{align} 

\item Next, we use  \texttt{ModuloMeas} to measure in the modulo magnetic quantum number basis. In particular, when the output state of the previous step is $\rho_1$, and the output state of this step is $\rho_2$, we set $(\rho_2,a) = \texttt{ModuloMeas}(\rho_1,r_\textsf{T})$. 
Here, $a = 0,\dots, r_\textsf{T}-1$ is what we call the syndrome outcome of the modulo measurement. 

\item Next, set $k=a+1$. Using this value of $k$, we apply a unitary $V_{\textsf{T},k}$ on $\rho'$, where 
\begin{align}
V_{\textsf{T},k} :
|{ j r_T + (k-1) -j_\textsf{T} })
_{ \textsf{T}}\>
\to
|{\rm code}_{j, \textsf{T}}\>,
\end{align}
for all $j=0,\dots,M-1$. 
The state is now $\rho_3 = V_{\textsf{T},k} 
\rho_2 V_{\textsf{T},k}^\dagger $, which is in a $\textsf{T}$-code.
\end{enumerate}
The final step is to map the $\textsf{T}$-code back to the original PI code, using ideas from the quantum Schur transform \cite{BCH-PRL-2006-schur,kirby2018practical,havlivcek2018quantum,krovi2019efficient,havlivcek2019classical,pearce2022multigraph} which is a unitary transformation which maps computational basis states to states of a superposition of states $|m_\textsf{T}\>$ over magnetic number $m$ and SYT $\textsf{T}$.
\newline

\noindent {\bf Algorithm shown in Fig.~\ref{fig:cgs}.}
\begin{enumerate}
\item 
The first step is to use a repeat-until-success procedure to project the state onto the symmetric subspace.
We achieve this with the following steps if the state is not in the symmetric subspace:
\begin{enumerate}
    \item Based on the SYT syndrome, we permute the qubits so that the qubits that correspond to the Young-Yamanouchi vector components with 1-bit values are rearranged as the last few qubits. We label these qubits as $\ell,\dots , N.$
    
    \item We perform a cascading sequence of $N-\ell+1$ inverse Clebsch-Gordan transformation to decouple the last $N-\ell +1 $ qubits. This reproduces part of the inverse quantum Schur transform in Ref.~\cite{BCH-PRL-2006-schur}. 
Here, each Clebsch-Gordan transformation has two CNOTs and a controlled-rotation.
    
    \item We measure the last $N-\ell +1 $ qubits in the computational basis.
    \item We measure the total angular momentum of the qubits in the sets $\{1,\dots, \ell\} , \dots, \{1,\dots, N\}$, and derive a new SYT. If we find that the state is not in the symmetric subspace, we repeat steps 1(a) to 1(d). If the state is in the symmetric subspace, we proceed to Step 2.
\end{enumerate}

\item 
We use all of the prior measurement outcomes in Step 1 to ascertain how the relative amplitudes of the PI code has shifted amongst the logical codewords. Then we apply an amplitude rebalancing algorithm to rebalance these amplitudes of the logical codewords. 
The amplitude rebalancing algorithm terminates with high probability in a constant number of projective measurement steps.

\end{enumerate}
When the logical codewords of the $\textsf{T}$-codes are normalized, this procedure recovers the original PI codes.

%%%%%%%%%%%%%%%%%%%%%
\iftoggle{science}
{\subsubsection{Teleportation based QEC}}
{\subsection{Teleportation based QEC}}
%%%%%%%%%%%%%%%%%%%%%
\label{ssec:teleporationQEC}

By consuming additional ancilla states, we can design a quantum error correction procedure for general errors that does not rely on the quantum Schur transform. The main idea is that we can bring our quantum state in the space $\mathcal P^{\textsf{T}}$ to the symmetric subspace by state teleportation that consumes a single logical ancilla state \cite[(7)]{ZLC00}. 
For this idea to work, we restrict ourselves to shifted gnu codes with odd $g$ and the following properties. 
\begin{enumerate}
    \item Define the unitary map 
    $X_{\rm schur} : |m_{\textsf{T}}\> \to |-m_{\textsf{T}}\>$ for every 
    ${m = -j_{\textsf{T}}, \dots, j_{\textsf{T}}}$ and SYT $\textsf{T}$.
    We need that $X_{\rm schur} $
    implements the logical $X$ gate on the $\textsf{T}$-codes for every correctible SYT $\textsf{T}$.
    \item 
    Define the unitary $\rm{C}_{\textsf{A}}{\rm X}_{\textsf{B}}$ on two registers $\textsf{A}$ and $\textsf{B}$ such that for any SYTs $\textsf{S}$ and $\textsf{T}$, the unitary $\rm{C}_{\textsf{A}}{\rm X}_{\textsf{B}}$ applies the map
    \begin{align}
   |l_{\textsf{S}}\>|m_{\textsf{T}}\> 
   \to 
|l_{\textsf{S}}\>   
X_{\rm schur}^{l}|m_{\textsf{T}}\> .
    \end{align}
\end{enumerate}
Here, registers $\textsf{A}$ and $\textsf{B}$ can potentially have different numbers of qubits. 

Next, we introduce a teleportation subroutine that takes an ancilla initialized as a shifted gnu code in register $\textsf{A}$ and the state $\rho$ that we want to teleport in register $\textsf{B}$. 
We assume that $\rho$ is a \textsf{T}-code that corresponds to some shifted gnu code.
The teleportation procedure entangles the registers and measures register $\textsf{B}$ to teleport the register in $\textsf{B}$ to the register in $\textsf{A}$.
We furthermore require that both shifted gnu codes encode a single logical qubit, and have the same gap $g$ that is also odd. We also need to be able to implement \texttt{ModuloMeas}.
\newline

\noindent{{\bf Procedure} \texttt{Teleportation}($\rho$):}
\begin{enumerate}
    \item Prepare $|+_L\>$ in a shifted gnu code in register $\textsf{A}$.
   \item Let $\rho$, a \textsf{T}-code, be in register $\textsf{B}$. 
\item Implement $\rm{C}_{\textsf{A}}{\rm X}_{\textsf{B}}$ on \textsf{A} and \textsf{B} respectively.
This effectively implements the logical CNOT with control on register \textsf{A} and target on register \textsf{B}.
\item 
On the state $\rho_\textsf{B}$ in register \textsf{B}, we implement Algorithm 2 (\texttt{ModuloMeas}).
Namely, we set 
$(\rho'_\textsf{B} , a) 
= \texttt{ModuloMeas}(\rho_\textsf{B}, 2g)$.
The output state is 
$\rho'_\textsf{A} \otimes \rho'_\textsf{B}$.
(This allows us to deduce the logical Z measurement outcome on the \textsf{T}-code in register \textsf{B}.)

\item 
Suppose that $\rho$ corresponds to a gnu code with shift $s$. Then set $\sigma = {\rm mod}(a + j_{\textsf T} -s, 2g)$.

\item
If $\sigma = 0, \dots, (g-1)/2$ or 
$\sigma = 2g - (g-1)/2, \dots, 2g-1$, then we declare that we obtained a logical 0 state on register \textsf{B}.
Then we do nothing to register \textsf{A}.

\item For all other values of $\sigma$, we declare that we obtained a logical 1 state on register \textsf{B}.
Then we apply $X_{\rm schur}$ on register $\textsf{A}$, implementing the logical $X$ correction on register \textsf{A}.
\end{enumerate} 
This teleportation subroutine which we depict in Figure \ref{fig:teleportation} is agnostic of which SYT syndome the state $\rho$ corresponds to, and how many qubits are in registers \textsf{A} and \textsf{B}.
We can use the teleportation subroutine to teleport \textsf{T}-codes to PI quantum codes.
Section \ref{sec:quantum-control} explains the physical implementation of steps in this teleportation subroutine.

\begin{figure}[t!]
\centering  
\includegraphics[width=0.5\textwidth]{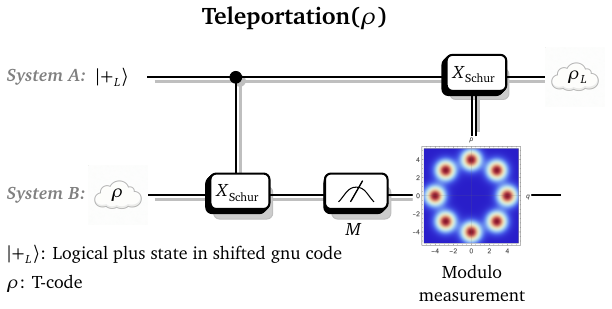} 
 \caption{
 % \yo{Refer to Figure 4.}
 {\bf Teleportation.}
 Schematic for the teleportation of a state $\rho$ that is a \textsf{T}-code to a state $\rho_L$ in a shifted gnu code, with the help of (1) an ancilla logical state $|+_L\>$ in the shifted gnu code, (2) a logical controlled-NOT gate $\rm{C}_{\textsf{A}}{\rm X}_{\textsf{B}}$ between code A and code B, (3) a modulo measurement, and (4) a conditional $X_{\rm schur}$ operation depending on the logical-$Z$ basis measurement outcome.
 }\label{fig:teleportation}
 \end{figure}

%
%For the purposes of correcting general errors, this teleportation subroutine takes as input a state supported on $\mathcal P^\textsf{T}$ in register $\textsf{B}$ for any SYT syndrome $\textsf{T}$, and teleports it to a state in the symmetric subspace of register $\textsf{A}$.
%More explicitly, any pure state in register $\textsf{B}$ with the decomposition
%$|\psi\> = \sum_{m} a_m|m_{\textsf{T}}\>$,
%the teleportation subroutine maps $|\psi\>$ to 
%$|\psi\> = \sum_{m} a_m|m_{\textsf{T'}}\> 
%= \sum_{m} a_m 
%|D^{N_{\textsf{A}}}_{m+N_{\textsf{A}}/2}\>$
%where $\textsf{T'}$ corresponds to the SYT that labels the symmetric subspace,
%and ${N_{\textsf{A}}}$ denotes the number of qubits in register $\textsf{A}$.

With this teleportation subroutine, 
we can also implement QEC for general errors
without employing quantum Schur transforms.
Namely, instead of using quantum Schur transforms to bring \textsf{T}-codes to the symmetric subspace, we can directly teleport $\textsf{T}$-codes to the original PI code.
% in Step 9 of Algorithm 3.

%The total number of control operations necessary for teleportation-based QEC is as follows. To obtain the SYT syndrome $\textsf{T}$, measurements of each of the $J_{[k]}^2$ operators for $1\leq k\leq N$ are required, each of which can be done with a single interaction round between the spins and the bosonic mode. Then the active steps of the teleportation procedure have the following cost: (1) $N$ GPGs for logical state preparation, (3) one GPG with linear coupling, (4) one GPG with quadratic coupling, (7) one transversal bit flip. The total number of control operations is then $2N+1$ interaction rounds between the spins and the bosonic mode and one transversal spin flip.

\iftoggle{science}
{\subsubsection{Simpler correction for deletion errors}  
}
{\subsubsection{Simpler correction for deletion errors}  
}
%%%%%%%%%%%%%%%%%%%%%%%%%%%%%%%%%%%%%%
%%%%%%%%%%%%%%%%%%%%%%%%%%%%%%%%%%%%%% 
 \label{ssec:correcting-deletions}
%%%%%%%%%%%%%%%%%%%%%

When deletion errors occur on PI quantum codes, the resultant state remains in the symmetric subspace, but on fewer qubits. If we accept having a quantum error corrected state with fewer qubits, we need not measure total angular momenta as we did in the previous section.
Rather, we can use a much simpler QEC scheme. 

For shifted gnu codes with $g$ larger than the number of deletions, the effect of deletions is to randomly shift the state into a linear combination of states supported on Dicke states of weights with different values modulo $g$. 
Then the error correction procedure is to first measure the Dicke weights of the states modulo $g$, then apply geometric phase gates to bring the state back into the codespace.

Now suppose that $t$ deletion errors occur on a PI code that corrects at least $t$ errors.
Then, an initially $N$-qubit pure state $|\psi\> = \alpha |0_L\> + \beta |1_L\>$ in the codespace 
becomes an $(N-t)$-qubit mixed state over subnormalized states 
$|\psi\>_a = \alpha |0_L\>_a + \beta |1_L\>_a$ \cite{ouyang2022finite}
where
\begin{align}
%|\psi\>_a = \sum_{w=0}^N a_w \delta[a \in A_w] |H^{N-t}_{w-a}\>.
|\psi\>_a = \sum_{w=a}^{N-t+a} a_w  
\frac{
\sqrt{\binom {N-t}{w-a}}
}
{
\sqrt{\binom {N}{w}}
}
|D^{N-t}_{w-a}\>.
\label{eq:psia-delete}
\end{align} 
Here, $a=0,\dots,t$ is the syndrome associated with the deletion error that quantifies the amount of shift in the Dicke weights.
When the PI code is furthermore supported on Dicke states with distinct weights at least $g$ apart with $g \ge t+1$, the states $|\psi\>_a$ and $|\psi\>_{a'}$ are orthogonal for distinct $a$ and $a'$.

Now, we will discuss performing QEC on the codespace without completely decoding the code, and thereby obtain a simpler QEC algorithm for deletion errors than that of Section \ref{sec:QEC}. 
%%%%%%%%%%
\begin{enumerate}
\item To obtain a particular subnormalized state 
$|\psi\>_a$ along with the syndrome $a$, we can 
set $\rho_a = |\psi\>_a\<\psi|_a/ \<\psi|_a |\psi\>_a$ and obtain $(\rho_a,a) = \texttt{ModuloMeas}(\rho, g)$.

\item 
We perform a non-unique unitary $V_a$ that maps the normalized versions of the states $|0_L\>_a$ and $|1_L\>_a$ to $|0_L\>$ and $|1_L\>$ respectively. 
For us, it suffices to implement $V_a$
%which acts only non-trivially on the symmetric subspace. A concrete way to perform this mapping 
using the $k=2$ dimensional subspace mapping described in Sec.\ref{subsec:stateprep}, which can be implemented using 
$\lceil 2(7Nk/3-1)\rceil$
linear GPGs and $\lceil 16 N/3\rceil$ transversal spin rotations.

\end{enumerate}
If the original PI code is a shifted gnu code with shift $s \ge t$, 
the resultant PI code is a shifted gnu code on $t$ fewer qubits, and with a shift $s - a$.

%%%%%%%%%%%%%%%%%%%%%%%%%%%%%%%%%%%%
\iftoggle{science}
{\subsection{Bosonic mode assisted QEC}
}
{\section{Bosonic mode assisted QEC}
}
\label{sec:quantum-control}
%%%%%%%%%%%%%%%%%%%%%%%%%%%%%%%%%%%%

Here, we describe how to implement QEC operations on symmetric states using quantum control of the spins together with a bosonic mode.

%%%%%%%%%%%%%%%%%%%%%%%%%%%%%%%%%%%%
\iftoggle{science}
{\subsubsection{Measurement of \texorpdfstring{$\hat{J}^2$}{J2}}
}
{\subsection{Measurement of \texorpdfstring{$\hat{J}^2$}{J2}}
}
\label{ssec:measure_total_angular_momentum}
%%%%%%%%%%%%%%%%%%%%%%%%%%%%%%%%%%%%
Measurement of total spin angular momentum involves measuring $\hat{J}^{x^2}+\hat{J}^{y^2}+\hat{J}^{z^2}=\hat{J}^2$ with eigenvalues $j(j+1)$. 
While $\hat{J}^2$ is clearly a Hermitian operator and in principle measurable, the actual physical construction of such an observable is not so straightforward. One possibility is to leverage recent results which show that singlet/triplet measurements on pairs of qubits can efficiently simulate universal quantum computation, or ${\bf STP=BQP}$ \cite{rudolph2023relational,Freedman2021symmetryprotected}. Once the number of pairwise qubit measurements becomes polynomially large compared to $J^2$, the accuracy converges exponentially.

Another way to 
measure $\hat J^2$ is to couple the spins to an ancillary bosonic mode and measure that. Consider the following Hamiltonian which generates a displacement of the position quadrature of the mode dependent on the total angular momentum
\begin{equation}
    \hat{H}=\hat{H}_0+\hat{H}_{\rm d}.
\end{equation}
Here, the free Hamiltonian of the mode and the spins is
\begin{equation}
    \hat{H}_0=\omega_c \hat{a}^{\dagger}\hat{a}+\omega_0 \hat{J}^z
    \label{freeH0}
\end{equation}
and the interaction term generating displacements is
\begin{equation}
    \hat{H}_{\rm d}=\xi_{\rm d}(\hat{a}^{\dagger}+\hat{a})\hat{J}^2,
    \label{Hdisp}
\end{equation}
where the mode creation and annihilation operators satisfy canonical commutation relations $[\hat{a},\hat{a}^{\dagger}]=1$.
In some physical systems it may be more natural for the mode to couple to the second moment of a particular component of spin, as in $\hat{H}'_{\rm d}=\xi_{\rm d}(\hat{a}^{\dagger}+\hat{a})\hat{J}^{z^2}$. In such a case one could approximate evolution by $\hat{H}_{\rm d}$ via a Trotterized expansion by a product of short time evolutions generated by $\hat{H}'_{\rm d}$ and conjugated by collective spin rotations around the $\hat{x}$ and $\hat{y}$ directions. 

In the simplest protocol, one begins with the mode prepared in the vacuum state, i.e. the coherent state $|\alpha=0\>$. Consider an initial spin state written as a superposition of basis states $\{|j,m\>\}$, where $j$ is the angular momentum, and $m$ is an eigenvalue of $\hat{J}^z$. After evolving the joint system according to $\hat{H}$ for a time $t$, the joint state is
\begin{equation}
\begin{array}{lll}
    e^{-i\hat{H}t}|\psi(0)\>&=&e^{-i\hat{H}t}\sum_{j,m}c_{j,m}
    |j,m\>\otimes |\alpha=0\>\\
    &=&\sum_{j,m}e^{-i \phi_j(t)}c_{j,m}|j,m\>|\alpha_j(t)\>,
    \end{array}
\end{equation}
where the combined dynamical and geometric phase is (see e.g. \cite{PhysRevB.105.094308} which also accounts for finite temperature and decay of the mode)
\[
\phi_{j,m}(t)=\omega_0 m t+\frac{\xi_{\rm d}^2}{\omega_c^2}(\omega_c t-\sin(\omega_c t))[j(j+1)]^2,
\]
and the displaced coherent state amplitude is
\[
\alpha_j(t)=-\frac{\xi_{\rm d}}{\omega_c}(1-e^{-i \omega_c t})j(j+1).
\]
The effect is the mode state becomes a displaced coherent state with the magnitude of displacement proportional to $j(j+1)$. The $\hat{q}$ mode quadrature can then be measured by homodyne detection. Choosing an interaction time satisfying $\tau=(2k+1)\pi/\omega_c $ for $k\in \mathbb{N}$, the minimum separation in phase space of the evolved coherent states amoung the different $\hat{J}$ eigenspaces is $\Delta \alpha_{\rm min}=|\alpha_1(\tau)-\alpha_0(\tau)|=4\xi_{\rm d}/\omega_c$. The overlap between these pointer states is $|\langle\alpha_0(\tau)|\alpha_1(\tau)\rangle|^2=e^{-\Delta\alpha_{\rm min}^2}$, hence for a nearly projective measurement it is required that $e^{-(4\xi_{\rm d}/\omega_c)^2}\ll 1$. If $\xi_{\rm d}/\omega_c$ is not large enough to satisfy this criterion, one could begin with a squeezed vacuum mode state with initial position variance $(\Delta \hat{X})^2=e^{-2r}/2$, in which case the criterion on distinguishability of the pointer states is
$e^{-(4e^r\xi_{\rm d }/\omega_c)^2}\ll 1$.
%In fact, thanks to the quadratic scaling of the eigenvalues of $\hat{J}^2$, if we assume that the initial spin state is dominated by components near maximum $j=n/2$, which is valid if only a constant number of amplitude damping or deletion errors occur, then the criterion for a projective measurement is $1/2\ll \gamma_{\rm d} t n$.

This primitive can be applied to compute the total angular momentum of any subset of spins and hence to project onto an SYT, which is required for correction of general errors. 

 \begin{figure}[h]
\centering  
\includegraphics[width=\columnwidth]{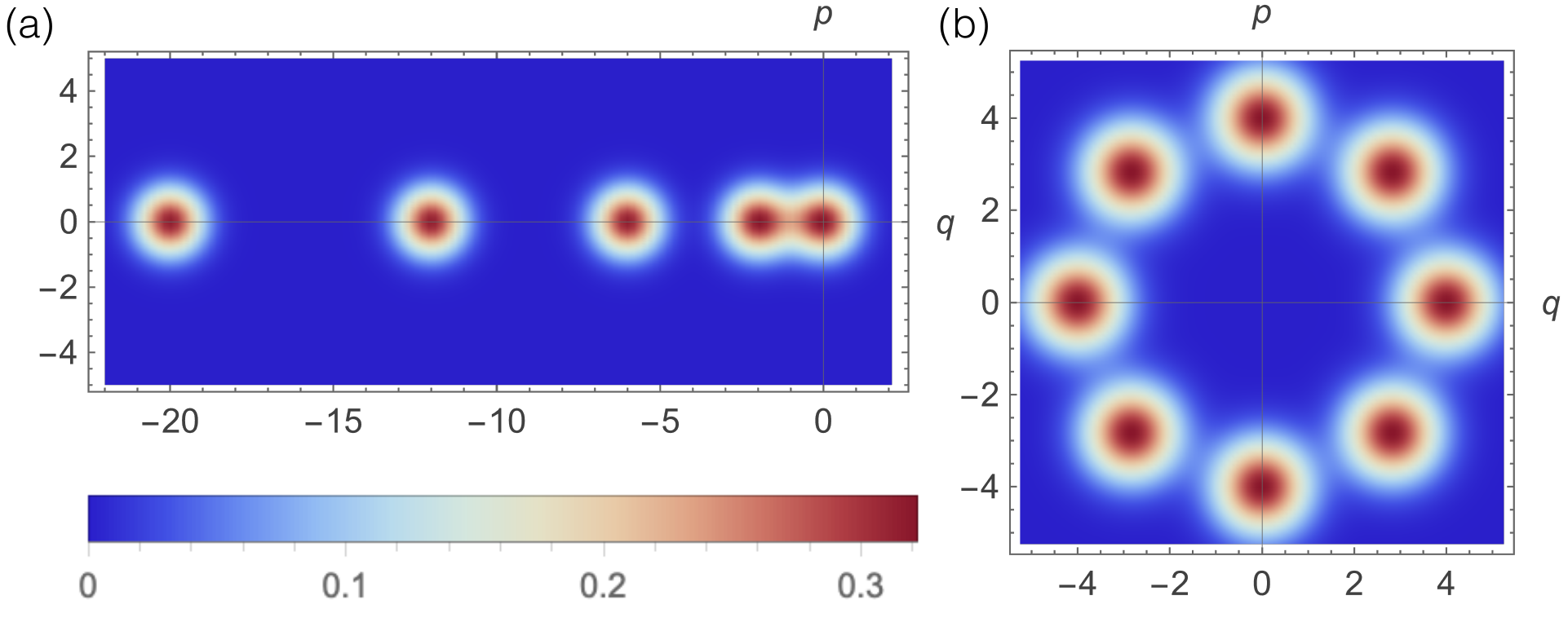} 
 \caption{{\bf Measurements on the bosonic mode.} Observables on spins mapped to quadratures of a bosonic mode as illustrated by plots of the Wigner distribution. (a) Measurement of $\hat{J}^2$. The mode is prepared in the vacuum $|\alpha=0\>$ and displaced by evolution generated by $\hat{H}_{\rm d}$ \eqref{Hdisp} over a time $\tau$ providing projection onto $j$ eigenspaces. Here $\omega_c\tau=(2k+1)\pi$ for $k\in \mathbb{N}$, $\xi_{\rm d}/\omega_c=1/2$, and $j\in\{0,1,2,3,4\}$. (b) Modular measurement of $\hat{J}^z$. The mode is prepared in a displaced vacuum state $|\alpha\>$ and rotated by evolution generated by $\hat{H}_{\rm r}$ \eqref{Hrot} providing projection onto Dicke weight $w\mod g$ eigenspaces. The evolution time is set to satisfy $\tau=2\pi/(g\xi_{\rm r})$. Here $g=8$ and $\alpha=4$.   
 }\label{fig:phasespace}
 \end{figure}

\iftoggle{science}
{\subsubsection{Modular measurement of \texorpdfstring{$\hat{J}^z$}{Jz}.}
}
{\subsection{Modular measurement of \texorpdfstring{$\hat{J}^z$}{Jz}.}
}
%\label{ssec:modulo}

In order to do the error detection for probe states using the gnu code, we need to perform a non-destructive measurement of the operator 
$\hat{J}^z+\frac{N}{2}{\bf 1}$ modulo 
$g\in \mathbb{Z}_{N+1}$.  
We focus on measurement of the Dicke weight w, in the symmetric subspace, but the argument follows for measurement in any SYT. Recall the Dicke space on $N$ spins as 
$\mathcal{H}={\rm span}_{\mathbb{C}}\{|D^N_w\>\}_{w=0}^{N}$. 

That is we want a projector onto the eigensubspaces $\{\mathcal{H}_p\}_{p=0}^{g-1}$ where $\mathcal{H}_p={\rm span}_{\mathbb{C}}\{|D^N_w\>;\ w\mod g=p\}$. 
One way to do this is to employ a dispersive interaction between a bosonic mode and the spins of the type
\begin{equation}
    \hat{H}_{\rm r}=\xi_{\rm r} \hat{J}^z \hat{a}^{\dagger}\hat{a}.
    \label{Hrot}
\end{equation} If we start with the bosonic mode prepared in the coherent state $|\alpha \>$, with $2\pi |\alpha|^2\gg g$, and evolve by the Hamiltonian $\hat{H}=\hat{H}_0+\hat{H}_{\rm r}$ for a time $t$
\begin{equation}
\begin{array}{lll}
    e^{-i\hat{H}t}|\psi(0)\>&=&e^{-i\hat{H}t}\sum_{j,m}c_{j,m}
    |j,\textsf{T},m\>\otimes |\alpha\>\\
    &=&\sum_{j,m}e^{-i \phi_m(t)}c_{j,m}|j,m\>|\alpha_m(t)\>,
    \end{array}
\end{equation}
where 
\[
\phi_m(t)=\omega_0 m t
\]
and
\[
\alpha_m(t)=e^{-it (\omega_c+m\xi_r) }.
\]
Starting with the mode in a displaced vacuum coherent state $|\alpha\rangle$, and evolving over a time
$\tau=2\pi/(g \xi_{\rm r})$, the mode state correlated with the spin state having Dicke weight $w$ is mapped to a coherent state rotated in phase space by an angle $\frac{2\pi}{g}\times (w\mod g)$ enabling a modular measurement of $\hat{J}^z$ by heterodyne measurement of mode. The minimum separation in phase space between the evolved coherent states is $\Delta \alpha_{\rm min}=|\alpha_{m}(\tau)-\alpha_{m\pm 1}(\tau)|=2|\alpha|\sin(\pi/g)$ and thus the condition for distinguishable pointer states is $e^{-4|\alpha|^2\sin^2(\pi/g)}\ll 1$. 

We have described the use of measurement of Dicke weight modulo $g$ for error detection where we don't want to disturb the logical information, but the same method works for {\bf Procedure} \texttt{Teleportation} where we do want to perform a projection onto logical states with the replacement $g\rightarrow 2g$. 
\newline

% \subsection{Shifting up the Dicke ladder}
% \label{ssec:shifting}
% Correction is not so straightforward as given an outcome $r$ from the above detection scheme we require an operation that injects exactly $r$ excitations into probe in a PI manner. Because of the equal energy splittings of the Dicke states, this can't be done directly using linear optical transformations. It would be possible by transferring an initial Fock state $|r\> $ prepared in the bosonic mode (resonant with the qubit energy splitting) and swapped into the probe using a Hamiltonian such as $\hat{H}_c=g_c (\hat{J}^{+})^r \hat{a}^r+h.c.$. However, this requires high nonlinearities for the mode preparation and interaction. Another possibility is to use interactions between the probe spins, e.g. dipole-dipole interactions in a ring geometry, to split the energies of the Dicke states in an $M$ dependent way and use frequency selective pulses to walk the probe states up the Dicke ladder \cite{Higgins2014}. However, such a scheme restricts the geometry of the setup, and also suffers errors from off resonant excitations and couplings outside the Dicke space induced by the dipole dipole interaction.
% A feasible scheme for correction is to use geometric control to perform a unitary extension of the mapping:$
% \{|D^N_{gw-r}\>\rightarrow |D^N_{gw}\>\}$. As shown in \cite{johnsson2020geometric}, arbitrary unitaries on the Dicke subspace can be synthesized using $O(N^2)$ geometric phase gates and global rotation gates. If only $s$ state mappings are non trivial then the complexity scales like $O(sN)$. 

\iftoggle{science}
{\subsubsection{State preparation of PI codes}}
{\subsection{State preparation of PI codes}}
\label{subsec:stateprep}

State preparation of any logical state of a PI code is equivalent to state synthesis in the Dicke subspace. There are now several protocols for exact \cite{johnsson2020geometric,PhysRevA.110.062610,k3bb-yfdv} and approximate \cite{PhysRevLett.132.153601} unitary and state synthesis in the Dicke subspace of $N$ qubits that employ coupling the qubits to a quantized bosonic mode. The work horse tool for quantum control is a composition of permutation invariant global rotations of the form $U=e^{i \theta \hat{J}^{\alpha}}$ about the axes $\alpha\in\{\hat{x},\hat{y},\hat{z}\}$ together with one or more types of geometric phase gates which are unitaries generated by a operator non-linear in the angular momentum component $\hat{J}^z$, with an action angle that depends on the area traversed in phase space by the mode mediating the interaction. The most commonly encountered GPG in experiments, used routinely now with trapped ions by coupling the electronic spin states to a motional mode of the ion string, is the M\o lmer-S\o rensen gate \cite{MSGate} $e^{i\phi \hat{J}^{z^2}}$. This gate has the advantage of being generated by coupling between the spins and the mode that is linear in the creation and annihilation operators and we will refer to such gates as linear GPGs. We shall see in Sec.\ref{ssec:teleport} that a dispersive GPG, obtained via a mode excitation number dependent coupling of spins to the mode, generates a higher non-linearity in $\hat{J}^z$ that is useful for teleportation. Beyond trapped ions, experiments have now demonstrated the essential steps needed to realize these GPS using: Rydberg atom qubits coupled to an optical mode of a Fabry-Perot cavity \cite{grinkemeyerErrordetectedQuantumOperations2025, hunger_fiber_2010-1, uphoff_frequency_2015} or utilizing circular Rydberg transitions coupled to a microwave cavity \cite{zhang2025opticallyaccessiblehighfinessemillimeterwave}. 

In Ref. \cite{ PhysRevResearch.7.L022072} it is shown that $\lceil{2N/3\rceil}$ linear GPGs and $\lceil{4N/3\rceil}$ transversal spin rotations suffice for exact state synthesis in the symmetric subspace, although in practice even a constant number of GPGs can approximate state preparation with high fidelity.
A noise optimized scheme for state and unitary synthesis in the Dicke space is given in Ref.\cite{k3bb-yfdv} where it is shown that for some states that have amplitudes sparse in the Dicke basis, high fidelity state synthesis can be achieved with a number of GPGs independent of $N$. The dispersive GPG can also be used for exact state preparation using at most $N-1$ GPGs and $N+1$ transversal spin rotations \cite{johnsson2020geometric}. 

Indeed this method can be extended to subspace mapping, i.e. the unitary synthesis of $k$ state mappings $\{|a_j\>\rightarrow|b_j\>\}_{j=1}^k$, where the set $S_a=\{|a_j\>\}$ is orthonormal as is $S_b=\{|b_j\>\}$, but the sets are not necessarily mutually orthogonal. As shown in Ref.~\cite{MBJD09}, it suffices to use $k$ instances of the following composition: a unitary state synthesis mapping, followed by the phasing on any symmetric state, e.g. the product state $|D^{N}_{N}\>$, followed by an inverse state mapping. Phasing of the Dicke state $|D^{N}_{N}\>$ is equivalent to the $C_{N-1}Z$ phase gate and is achieved using $N-1$ linear GPGs \cite{PhysRevA.110.062610}. Hence the overall cost of $k$ dimensional subspace mapping is 
$\lceil k(7Nk/3-1)\rceil$
linear GPGs and $\lceil 8 Nk/3\rceil$ transversal spin rotations. In the case $k=N+1$ we have full unitary synthesis over the Dicke subspace.

\iftoggle{science}
{\subsubsection{Operations for teleportation}}
{\subsection{Operations for teleportation}}
\label{ssec:teleport}

%^\GKB{The requisite operations for teleportation are described in Sec. \ref{ssec:teleportationQEC}. 
Since we have already described the resources required to prepare the state $|+_L\>$ and to perform the modular measurement of $\hat{J}^z$, the remaining requirement for teleporation based correction is to implement the controlled not gate
$\rm{C}_{\textsf{A}}{\rm X}_{\textsf{B}}$ between an ancilla system {\textsf{A}} and the register system {\textsf{B}}, the latter of which has undergone error and been projected onto a SYT ${\textsf{T}}$. 
Recall the action of $X_{\rm schur}$ on an eigenstate of $\hat{J}^z$ in an SYT labeled ${\textsf{T}}$ is to reverse the sign on the eigenstates of $\hat{J}^z$. We use the following definition, which is a bit flip up to a phase $i$:
\[
X_{\rm schur}|m_{\textsf{T}}\>=i|-m_{\textsf{T}}\>.
\]
This is a transversal bit flip operation since for $X_{\rm schur}=i\prod_{j=1}^M X_j$, then
$X_{\rm schur}^{\dagger}\hat{J}^zX_{\rm schur}=-\hat{J}^z$.
The definition with the phase $i$ is convenient for the boson mediated gate below and the phase can be corrected for as a Clifford gate correction in subsequent operations or by instead performing state preparation on system $\textsf{A}$ to the target $\frac{1}{\sqrt{2}}(|0_L\rangle-i|1_L\rangle)$.

%\label{ssec:GPG}
The controlled not gate
 can be realized using a dispersive GPG  \cite{PhysRevA.65.032327}. The method to do so is fully described in Appendix G2 of Ref.\cite{ouyang2025measurement} where it is shown how to implement the operation $i\prod_{j\in B}Z_j$ on a target register conditioned on the control register being the state $|1_L\rangle$ of a PI code. Because this product operator is equal to $X_{\rm schur}$ up to conjugation by a transversal spin rotation, the following sequence suffices: 
\begin{equation}
{\rm{C}_{\textsf{A}}{\rm X}_{\textsf{B}}} = e^{-i\frac{\pi}{2} \hat{J}_B^y} 
\Lambda_A(U_{\rm nl}-{\rm GPG}(\theta,\phi,\chi))e^{i\frac{\pi}{2} \hat{J}_B^y},  
\end{equation}
where
\begin{equation}
\begin{array}{lll}
\Lambda_A(U_{\rm nl}-{\rm GPG}(\theta,\phi,\chi))&=&\Lambda_A(-\beta,\delta)R(\theta \hat{J}_B^z)\\
&&\Lambda_A(-\alpha,\delta)R(-\theta \hat{J}_B^z)\\
&&\Lambda_A(\beta,\delta)R(\theta \hat{J}_B^z)\Lambda_A(\alpha,\delta)\\
&&R(-\theta \hat{J}_B^z).
\end{array}
\end{equation}
 Here $\chi=|\alpha\beta||1-e^{2 i \delta}|$ and $\phi=\arg(\alpha)-\arg(\beta)$. The conditional rotational operator generated by a dispersive coupling to the mode is
$ R(\theta \hat{J}_B^z)=e^{i\theta \hat{J}^z_B \hat{a}^{\dagger}\hat{a}}$,
and the conditional mode rotation is 
\begin{equation}
\Lambda_A(\alpha,\delta)=D(\alpha)R(\theta \hat{J}_A^z) D(-\alpha)R(-\theta \hat{J}_A^z), 
\end{equation}
where the mode displacement operator is
$D(\alpha)=e^{\alpha a^{\dagger}-\alpha^{\ast}a}$. The specific choices for the parameters in these gates depends on the PI code used (see \cite{ouyang2025measurement} for details). For example, to teleport between $7$ qubit PI codes, where the $|0_L\rangle$ and $|1_L\rangle$ states are a superposition of Dicke states with Hamming weight $0\mod q$ and $s\mod q$ respectively, with $s=2,q=5$, one chooses $\delta=s\pi/q$, $\phi=0$, $\theta=\pi$, and $\alpha=\sqrt{\pi}/(4|\sin(\delta)|)$. In total, the number of elementary spin-mode coupling gates $R(\theta \hat{J}_B^z)$ is $12$ and the number of transverse spin rotations is $2$.

\iftoggle{science}
{\subsection{Complexity}}
{\section{Complexity}}

We can count the cost in elementary control operations to perform QEC on PI codes using either the teleporation based or gate based approach. Both require measurement of the SYT syndrome $\textsf{T}$ which involves $N$ measurements of total spin angular momentum on subsets of spins. The operators commute but the physical realization described in Sec.\ref{ssec:measure_total_angular_momentum} takes $N-1$ steps. 
 
%The total number of control operations necessary for teleportation-based QEC is as follows. %To obtain the SYT syndrome $\textsf{T}$, measurements of each of the $J_{[k]}^2$ operators for $1\leq k\leq N$ are required, each of which can be done with a single interaction round between the spins and the bosonic mode. 
The active steps of {\bf Procedure} \texttt{Teleportation}, after measurement of the $\textsf{T}$ syndrome, have the following cost: $\lceil{2N/3\rceil}$ linear GPGs and $\lceil{4N/3\rceil}$ transveral spin rotations for logical state preparation, one dispersive GPG and two transversal spin rotations  for the $\rm{C}_{\textsf{A}}{\rm X}_{\textsf{B}}$ gate, one instance of \texttt{ModuloMeas}, and finally one transversal bit flip for $X_{\rm schur}$. The number of elementary operations required is then: $\lceil{2N/3\rceil}$ linear GPGs, one dispersive GPG (consisting of $12$ dispersive spin-mode coupling gates), $\lceil{4N/3\rceil}+3$ transversal spin rotations, and one measurement that can be performed using a dispersive spin coupling with heterodyne measurement. Only  addressability on the entire code register A vs. B during transversal spin rotations and coupling to the common bosonic mode is required. 

The gate based approach uses part of the inverse quantum Schur transform for decoding. To perform this algorithm on $N$ qubits to an accuracy of $\epsilon$, it suffices to use $O(t \log(N/\epsilon))$ gates \cite{burchardt2025highdimensionalquantumschurtransforms}, where $t$ is the number of boxes in the second row of the SYT syndrome.
Typically for PI codes, $t$ is sublinear in $N$, and hence the inverse Schur transform contribution to the complexity of the decoding circuit is sublinear in $N$. 
\section{Discussion} 
   
Permutation invariant qubit codes have appealing features that arise intrinsically due to their symmetry under qubit relabeling, like the fact that error locations do not need to be tracked to be corrected. This symmetry makes the method for correction very different from stabilizer codes which rely on measuring parity constraints on qubits. 
%PI codes are a promising family of quantum error correction codes because their permutational symmetry appears advantageous for near-term implementation. 
%An important step that would bring PI codes closer to physical realization is the design of explicit quantum error correction protocols for them. 
%However, since PI codes are generally not stabilizer codes, the usual quantum error correction formalism does not allow us to derive quantum error correction protocols for them.  
By leveraging the representation theory of the symmetric group, we show that the appropriate measurement needed is onto a symmetric young tableau to isolate how the state left the symmetric subspace under an error together with modular measurement of the Dicke weights which allows to infer how excitations were added or subtracted from the state without revealing logical information. All the steps for syndrome extraction and correction can be done in a number of elementary steps that is linear in the number of qubits $N$. By expanding the gate library beyond the usual single and two qubit gate sets to include geometric phase gates, some of the QEC steps are simplified and require less addressability thus avoiding associated cross-talk errors. 
  
% For the correction of general errors in a PI code, we introduce a general method to measure and decode the errors. For this, we can use the full theory of syndrome extraction and decoding for arbitrary correctible errors on arbitrary PI codes.  

\section{Acknowledgements}\label{eq:acknow}
Y.O. acknowledges support from EPSRC Grant No. EP/W028115/1 and also the EPSRC funded QCI3 Hub under Grant No. EP/Z53318X/1.  
G.K.B. acknowledges support from the Australian Research Council Centre of Excellence for Engineered Quantum Systems (Grant No. CE 170100009).

\iftoggle{science}
{\section{Materials and Methods}
\subsection{Young diagrams and Young tableau}
\label{sssec:YDYT}
%%%%%%%%%%%%%%%%%%%%%%%%%%%%%%%
Consider Young diagrams \cite[Page 29]{stanley-enumerative-v1} comprising of $N$ boxes arranged in two left-justified rows.
We restrict our attention to Young diagrams with $r_1$ boxes on the first row and $r_2$ boxes on the second row where $r_1 \ge r_2$ and $r_1 + r_2 = N$. 
\begin{align}
\includegraphics[width=0.5\textwidth]{graphics/YoungTableau.pdf} 
\label{eq:YT}
\end{align}
In \eqref{eq:YT}(a), we depict a Young diagram with four boxes on the first row, and two boxes on the second row. 
Each Young diagram corresponds to an integer partition of $N$ with two parts.
A standard Young Tableau (SYT) is obtained by filling up the $N$ boxes in a Young diagram with integers from 1 to $N$ such that the integers strictly increase from left to right, and from top to bottom.
Given a Young diagram, for instance in \eqref{eq:YT}(a), we give two examples of SYTs that can be obtained in \eqref{eq:YT}(b) and \eqref{eq:YT}(c) respectively.
We can enumerate the number of SYTs consistent with a given Young diagram using the hook-length formula \cite[Corollary 7.21.6]{stanley-enumerative-v2}.
The hook-length formula states that the number of SYTs consistent with any Young diagram with $N$ boxes is equal to $N!$ divided by the hook-length of each box. The hook-length of a given box is the total number of boxes in its hook, where the box's hook includes the box itself and all other boxes to its right and bottom. 
\begin{align}
\includegraphics[width=0.5\textwidth]{graphics/HookLength_diagrams.pdf} 
\label{eq:hooklengths}
\end{align}
In \eqref{eq:hooklengths}, we shade the hooks of the labelled cells. 
In \eqref{eq:hooklengths}(a), (b) and (c), the hook-lengths are five, four and two respectively.
From the hook-length formula, the number of SYTs corresponding to a Young diagram with two rows is
\begin{align}
&\frac{N!}{(N-r_1)! (r_1+1)!/ (r_1-N+r_1+1)}\notag\\
=& \binom N{r_1} \frac{2r_1 - N + 1}{r_1+1}.\label{eq:num_SYT}
\end{align}

We can also represent SYTs as a binary vector, using the so-called Young-Yamanouchi basis. 
The binary vector $(x_1,\dots, x_N)$
is such that $x_j=1$ if the symbol $j$ in the SYT is in the second row, and $x_j=0$ if the symbol $j$ in the SYT is in the first row. Operationally, $x_j=0$ if the measured total angular momentum increases, and $x_j=1$ if the measured total angular momentum decreases.

Apart from SYTs, we also consider semistandard Young tableau (SSYT), where boxes are filled with integers that strictly increase from top to bottom and are non-decreasing from left to right. 
Here, we restrict our attention to SSYTs filled with the numbers 1 and 2. 
The number of such SSYTs obtainable from a Young diagram with two rows is 
\begin{align}
r_1-r_2+1.\label{eq:num_SSYT}
\end{align} 
We list all SSYTs filled with the numbers 1 and 2 when $r_1 = 4$ and $r_2=2$ in \eqref{eq:SSYTs} below.
%%%%%%%%%%%%%
 \begin{align}
\includegraphics[width=0.5\textwidth]{graphics/semistandard_YT.pdf} 
\label{eq:SSYTs}
 \end{align}
%%%%%%%%%%%%% 
From Schur-Weyl duality \cite[Chapter 9.1]{goodman2000representations} applied to quantum information theory \cite{BCH-PRL-2006-schur,harrow2013church}, we know that 
the $N$-qubit space $(\mathbb C^2)^{\otimes N}$ is isomorphic to 
\begin{align}
\bigoplus_{D} \mathcal Q^D \otimes \mathcal P^D, \label{eq:schurweyl}
\end{align}
where each $D$ denotes a Young diagram with $N$ boxes and two rows.
For every Young diagram $D$, $\mathcal Q^D$ is a space with basis elements labelled by all possible SYTs filled with integers 1 to $N$, and $\mathcal P^D$ is a space with basis elements labelled by all possible SSYTs filled with integers 1 and 2.
Since the basis of $\mathcal Q^D$ is labeled by SYTs of shape $D$, we can write $\mathcal Q^D = {\rm span}\{|\textsf{T}\> : \textsf{T} \mbox{ is a SYT for } D\}$.

Since $D$ always has two rows, we can represent it with the tuple $(r_1, r_2)$ where $r_1$ and $r_2$ count the number of boxes in the first and second rows respectively.
When $D=(N,0)$, we can see from \eqref{eq:num_SYT} and \eqref{eq:num_SSYT} that the dimension of $\mathcal Q^{(N,0)}$ is one, and the dimension of $\mathcal P^{(N,0)}$ is $N+1$. In fact, $\mathcal P^{(N,0)}$ corresponds precisely to the symmetric subspace of an $N$-qubit symmetric state.

\subsection{Symmetrizing lemma}

When a channel acts identically and independently on every qubit, it 
always maps a pure symmetric state to a density matrix that is block diagonal on the spaces $\mathcal Q^D \otimes \mathcal P^D$ as given \eqref{eq:schurweyl}.
For such block diagonal states, QEC can proceed by projecting the density matrix onto one of the blocks labelled by $D$.

However the density matrices that we encounter in the QEC of PI codes do not necessarily have this block diagonal structure.
In this scenario, such density matrices can always be made block diagonal in the Schur-Weyl basis by applying symmetrizing operations described by the quantum channel $\mathcal S$ that randomly permutes qubits, and has Kraus operators $\{\frac{1}{\sqrt N!} P_{\sigma} : \sigma \in S_N \}$,
where $P_\sigma$ denotes an $N$-qubit matrix representation of a permutation operator that permutes qubit labels according to the permutation $\sigma$.
In the following lemma, we prove that if we apply the symmetrizing channel, originally correctible errors remain correctible. 
%%%%%%%%%%%%%%%%%%
\begin{lemma}[Symmetrizing lemma]
\label{lem:sym}
Let $\mathcal C$ be any $N$-qubit PI code of distance $d$. Let $\mathcal N$ be any quantum channel with Kraus operators $K$ of weight at most $t$.
Then if $d \ge 2t+1$, the channels $\mathcal N$ and $\mathcal S \circ \mathcal N$ are both correctible with respect to $\mathcal C$.
\end{lemma}
%%%%%%%%%%%%%%%%%%
%%%%%%%%%%%%%%%%%%
\begin{proof}
The channel $\mathcal N$ is correctible, and hence satisfies the fundamental quantum error correction criterion \cite{KnL97}. To prove the lemma, we must show that the channel $\mathcal S \circ \mathcal N$ is also correctible.

%Now let $K$ denote a set of Kraus operators for the quantum channel $\mathcal N$.
Now, denote $\bar K = \{\frac{1}{\sqrt N!} P_{\sigma} A : \sigma \in S_N , A \in K \}$ as a set of Kraus operators for the quantum channel $\mathcal S \circ \mathcal N$. 
The Kraus operators in $\bar K$ are correctible if and only if for every $\sigma, \tau \in S_N$ and $A,B \in K$, there exists a $g_{A,B,\sigma,\tau} \in \mathbb C$ such that 
\begin{align}
\Pi  A^\dagger P_{\sigma}^\dagger  P_\tau B  \Pi
= g_{A,B,\sigma, \tau} \Pi ,
\label{qecc_of_SN}
\end{align}
where $\Pi$ the code projector for $\mathcal C$.

Since $\Pi$ is a projector onto the symmetric subspace, for all $\sigma \in S_n$,
we have that $P_\sigma \Pi = \Pi P_\sigma = \Pi.$
Denoting $A_\sigma = P_{\sigma} A P_{\sigma}^\dagger $
and $B_\sigma = P_{\tau} B P_{\tau}^\dagger$, note that \eqref{qecc_of_SN} is equivalent to 
\begin{align}
\Pi  A_\sigma^\dagger  B_\tau  \Pi = g_{A,B,\sigma, \tau}\Pi .
\label{qecc_of_SN2}
\end{align}
Since $A,B$ are operators with weight at most $t$, 
$A_\sigma$ and $B_\tau$ must also be operators of weight at most $t$. 
Hence, both $A_\sigma$ and $B_\tau$ are linear combinations of Pauli operators of weight at most $t$. 
Namely, 
\begin{align}
A_\sigma = \sum_{P : |P|\le t } a_{A,P,\sigma} P,
\quad
B_\tau     = \sum_{P : |Q|\le t } b_{B,Q,\sigma} P,
\end{align}
where $a_P$ and $a_Q$ are real coefficients. 
From this, it follows that the left side of \eqref{qecc_of_SN2} is equivalent to 
\begin{align}
\Pi  A_\sigma^\dagger  B_\tau  \Pi  = \sum_{P,Q} a_{A,P,\sigma}^* b_{B,Q,\sigma}  \Pi  P Q \Pi.
\label{qecc_of_SN3}
\end{align}
From the Knill-Laflamme condition \cite{KnL97},
since $\mathcal C$ is a code of distance $d$, for every Pauli $P$ and $Q$ of weight at most $t$, there exists a $c_{P,Q} \in \mathbb C$ such that 
\begin{align}
\Pi  P Q \Pi = c_{P,Q} \Pi.
\label{distance-KL}
\end{align}
Substituting \eqref{distance-KL} into the left side of \eqref{qecc_of_SN3},
we can conclude that 
\begin{align}
\Pi  A_\sigma^\dagger  B_\tau  \Pi 
= \sum_{P,Q} a_{A,P,\sigma}^* b_{B,Q,\sigma}  c_{P,Q}\Pi
\label{qecc_of_SN4}.
\end{align}
This implies that 
$\Pi  A^\dagger P_{\sigma}^\dagger  P_\tau B  \Pi
= g_{A,B,\sigma, \tau} \Pi ,$
where $
% \begin{align}
g_{A,B,\sigma, \tau} = \sum_{P,Q} a_{A,P,\sigma}^* b_{B,Q,\sigma}  c_{P,Q},
% \end{align}
$
and this proves that $\mathcal S \circ \mathcal N$ is also correctible with respect to the code $\mathcal C$.
\end{proof}
Lemma \ref{lem:sym} hints how we can perform quantum error correction on any PI quantum code.
Namely, if any correctible channel $\mathcal N$ introduces errors on a PI quantum code, 
we can project the state into the Schur-Weyl basis and still be able to correct the resultant errors. 
This is because a symmetrizing channel $\mathcal S$ makes a quantum state block-diagonal in the Schur-Weyl basis, and Lemma \ref{lem:sym} tells us that if $\mathcal N$ is correctible, $\mathcal S \circ \mathcal N$ is also correctible. 
Note that we do not have to actually physically symmetrize the channel, if we randomize the labels in our minds.
Section \ref{sssec:syndrome-extraction} illustrates how we may project the state onto the diagonal blocks in the Schur-Weyl basis.
%%%%%%%%%%%%%%%%%%
}
{}
%\input{methods-text}
%
%\appendix
%\input{SI_material}

\bibliography{ref}{}

\end{document}